 \documentclass{article}
\usepackage{arxiv}                                    
\usepackage{graphicx}          
\usepackage{xcolor}
\usepackage[utf8]{inputenc}
\usepackage{amsmath} 
\usepackage{amssymb}   
\usepackage{amsfonts}   
\usepackage{graphicx} 
\usepackage{enumitem} 
\usepackage{subcaption}

 \usepackage{algorithm}
 \usepackage{algpseudocode}
\newtheorem{assumption}{Assumption}

\newtheorem{theorem}{Theorem}

\newtheorem{remark}{Remark}
\newtheorem{proposition}{Proposition}

\usepackage{float}
\usepackage{booktabs}
\usepackage{siunitx}
\usepackage{cite}

\usepackage{threeparttable}
\usepackage{tabularx}
\usepackage{adjustbox} 
\usepackage{colortbl}   
\usepackage{makecell}

\newcommand*{\QEDA}{\hfill\ensuremath{\blacksquare}}

\title{MPC as a Copilot: A Predictive Filter Framework with Safety and Stability Guarantees} %

\author{ Yunda Yan\\
	Department of Computer Science\\
	University College London, London\\
    WC1E 6BT, UK \\
	\texttt{yunda.yan@ucl.ac.uk} \\
	\And Chenxi Tao\\
	School of Automation\\
    Southeast University, Nanjing\\
    210096, China\\
	\texttt{chenxi.tao@seu.edu.cn} \\
    \And Jinya Su\\
	School of Automation\\
    Southeast University, Nanjing\\
    210096, China\\
	\texttt{sucas@seu.edu.cn} \\
    \And Cunjia Liu\\
	Department of Aeronautical and Automotive Engineering\\
    Loughborough University, Leicestershire\\
    LE11 3TU, UK\\
	\texttt{c.liu5@lboro.ac.uk} \\
        \And Shihua Li\\
	School of Automation\\
    Southeast University, Nanjing\\
    210096, China\\
	\texttt{lsh@seu.edu.cn} \\
}

\begin{document}
\maketitle

\begin{abstract}
Ensuring both safety and stability remains a fundamental challenge in learning-based control, where goal-oriented policies often neglect system constraints and closed-loop state convergence. To address this limitation, this paper introduces the \textit{Predictive Safety--Stability Filter} (PS\textsuperscript{2}F), a unified predictive filter framework that guarantees constraint satisfaction and asymptotic stability within a single architecture. The PS\textsuperscript{2}F framework comprises two cascaded optimal control problems: a nominal model predictive control (MPC) layer that serves solely as a \textit{copilot}, implicitly defining a Lyapunov function and generating safety- and stability-certified predicted trajectories, and a secondary filtering layer that adjusts external command to remain within a provably safe and stable region. This cascaded structure enables PS\textsuperscript{2}F to inherit the theoretical guarantees of nominal MPC while accommodating goal-oriented external commands.  Rigorous analysis establishes recursive feasibility and asymptotic stability of the closed-loop system without introducing additional conservatism beyond that associated with the nominal MPC. Furthermore, a time-varying parameterisation allows PS\textsuperscript{2}F to transition smoothly between safety-prioritised and stability-oriented operation modes, providing a principled mechanism for balancing exploration and exploitation. The effectiveness of the proposed framework is demonstrated through comparative numerical experiments.
\end{abstract}

\keywords{                 
Predictive safety–stability filter; Model predictive control;   Safety and stability; Learning-based control}

\section{Introduction}

Recent advances in machine learning and high-fidelity simulators, coupled with increases in computational power and hardware efficiency, have sparked growing interest in learning-based and data-driven control methods~\cite{hewing2020learning}. Although these methods have demonstrated impressive performance in various real-world applications, such as aerial and ground vehicles~\cite{kaufmann2023champion, song2023reaching}, their theoretical foundations remain relatively underdeveloped, with only a limited number of results providing rigorous guarantees~\cite{berberich2020data,fazlyab2020safety}. This lack of formal guarantees presents significant risks in safety-critical domains such as robotics and automated driving, where violations of safety constraints can lead to catastrophic consequences~\cite{yan2023surviving, hsu2023safety}.

To mitigate this research gap, filter-like frameworks have emerged as a promising paradigm, serving as a safety layer that monitors and, whenever necessary, adjusts control inputs to ensure constraint satisfaction~\cite{wabersich2023data}. Most existing approaches rely on the construction of a control invariant set, either explicitly or implicitly, that guarantees safe system operation when subject to learning-based control inputs. For instance, the control barrier function (CBF) framework~\cite{ames2016control} employs a barrier function to characterise a safe set and enforces its invariance through a differential inequality. By contrast, the predictive safety filter (PSF)~\cite{wabersich2022predictive} leverages a constrained optimisation problem inspired by model predictive control (MPC), typically incorporating a terminal constraint to guarantee recursive feasibility and, thereby, safety. Both approaches have demonstrated remarkable success in ensuring safety for robotic systems~\cite{ha2025learning, tearle2021predictive}.

However, safety alone may not be sufficient. Persistent oscillations within the constraints can still degrade performance and undermine reliability. Embedding asymptotic stability into the safety framework is therefore crucial to ensure both constraint satisfaction and convergence, yet only a few attempts have been made in this direction. Within the CBF framework, a control Lyapunov function (CLF) can be incorporated into the optimisation problem to enforce stability; however, this often leads to infeasibility due to the absence of a unified CLF design~\cite{ames2016control}. Introducing a slack variable in the CLF constraint can alleviate this issue, but stability still cannot be guaranteed because the slack variable cannot, in general, be forced to remain zero at all times~\cite{jankovic2018robust}. In the absence of an embedded CLF, the external control law is generally required to stabilise to ensure both asymptotic stability and constraint satisfaction~\cite{cortez2022compatibility}. In \cite{zinage2025universal}, a universal barrier function is proposed to guarantee safety and stability simultaneously. However, this approach requires the nominal controller to be expressed in an integral form and is therefore not directly applicable to general learning-based or data-driven control settings. Within the PSF framework, stability results are even fewer. In~\cite{milios2024stability, didier2024predictive}, an additional stability constraint can be embedded into the original optimisation problem to ensure asymptotic stability; however, this modification may compromise the recursive feasibility property of the original formulation. Guaranteeing both safety and asymptotic stability without excessive conservatism is therefore already challenging, and an even more theoretically significant and practically meaningful question is whether the user can manually determine when the system should begin to stabilise. This capability is particularly valuable in exploration–exploitation scenarios in robotics~\cite{rickert2014balancing}, where safety must be maintained during the exploration or learning phase while asymptotic stability is enforced once reliable system knowledge becomes available. Shared autonomy~\cite{marcano2020review} and human–robot interaction~\cite{selvaggio2021autonomy} exemplify such settings: user-generated commands may be suboptimal or even unsafe, requiring the controller to grant operators sufficient freedom to express intent (safe exploration) while ensuring that all executed actions remain safe and ultimately drive the system back to a stable operating regime (stable exploitation).

Motivated by the aforementioned research challenges and gaps, we propose the 
\textit{Predictive Safety--Stability Filter} (PS\textsuperscript{2}F), a unified 
predictive control framework consisting of two optimal control problems (OCPs) 
arranged in a cascaded structure. The first OCP is a nominal MPC problem whose value 
function implicitly serves as a Lyapunov function for the nonlinear constrained 
system. Its optimal input and state trajectories represent the ideal, 
theoretically certified solutions that guarantee both safety and stability; however, 
these inputs are not applied directly to the system, but instead function as a 
\textit{copilot}. The external controller acts as a \textit{pilot}, generating 
goal-oriented commands that pursue high-level tasks or learning objectives. The 
second OCP then reconciles these two objectives by computing a control input that 
remains as close as possible to the external command while ensuring it lies within a 
provably safe and stable domain, as illustrated in Fig.~\ref{fig:ps2f_copilot}.  This layered structure enables PS\textsuperscript{2}F to inherit the recursive 
feasibility and stability guarantees of nominal MPC while preserving the flexibility 
needed to accommodate arbitrary external commands. Importantly, the second OCP is 
constructed so as not to introduce additional conservatism: feasibility is preserved 
as long as the initial state lies within the nominal MPC’s region of attraction. 
Furthermore, by appropriately designing time-varying parameters, 
PS\textsuperscript{2}F can prioritise safety during an initial phase and then 
deliberately transition into a stability-enforcing mode at a user-specified time. 
The ability of PS\textsuperscript{2}F to mediate between arbitrary external 
commands and safety–stability–certified feedback makes it particularly well suited 
for human-in-the-loop applications such as shared autonomy and human–robot 
interaction, where operators must retain freedom of intent while the system remains 
safe and ultimately convergent.

\begin{figure}[t]
    \centering
    \includegraphics[width=0.5\linewidth]{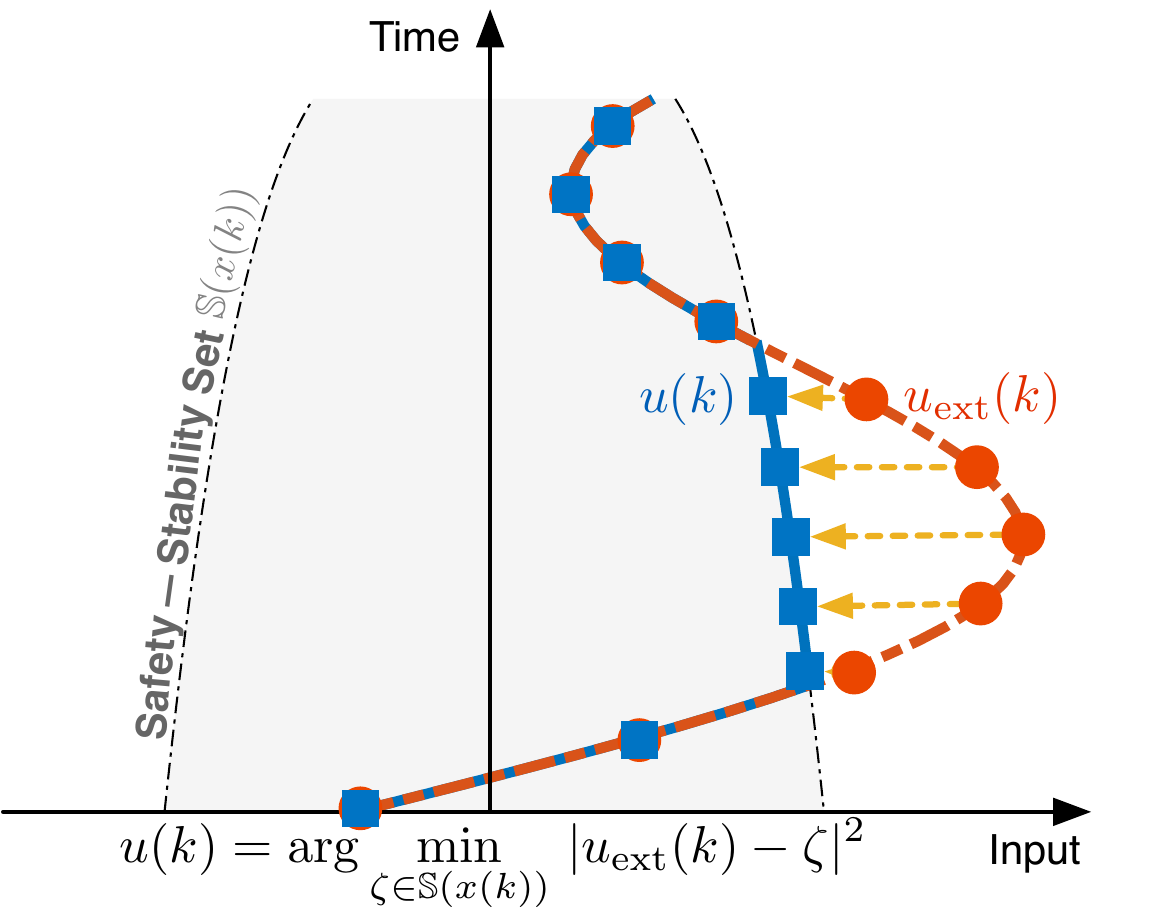}
    \caption{Conceptual illustration of the proposed PS\textsuperscript{2}F framework. 
    The external controller generates a goal-oriented command $u_{\mathrm{ext}}$, 
    which may prioritise efficiency or high-level task performance but can potentially violate constraints or compromise stability. 
    The PS\textsuperscript{2}F framework acts as a predictive filter that projects $u_{\mathrm{ext}}$ 
    onto the safety--stability set $\mathbb{S}(x)$ (defined in~\eqref{eq:admissible_controls_sim} and implicitly constructed via the two OCPs), 
    ensuring that the applied control input remains both safe and stable.}
    \label{fig:ps2f_copilot}
\end{figure}

The remainder of this paper is organised as follows.
Section 2 revisits the nominal MPC formulation and the conditions required to ensure recursive feasibility and asymptotic stability.
Section 3 presents the proposed PS\textsuperscript{2}F framework and provides a detailed theoretical analysis.
Section 4 demonstrates the effectiveness of the proposed method through comparative numerical studies.
Finally, Section 5 concludes the paper. For clarity of presentation, the main proofs are provided in Appendix A and Appendix B.






  

\textit{Notation:} $\mathbb{I}$ and $\mathbb{R}$ denote the sets of integers and real numbers, respectively, where superscripts or subscripts may be added to specify particular ranges (e.g., $\mathbb{I}_{\ge 0}$). The matrices $0_{n\times m} \in \mathbb{R}^{n\times m}$ and $I_n \in \mathbb{R}^{n\times n}$ denote the zero matrix and the identity matrix, respectively. For any sequence $\mathbf{u} = \{u_0, u_1, \dots, u_N\}$, we denote by $\mathbf{u}_{a:b}:= \{u_{a}, \dots, u_{b}\}$ the subsequence from index $a$ to $b$, where $a$ and $b$ are integers satisfying $a \le b$ and $a, b\in\mathbb{I}_{0:N}$. For any vector $x \in \mathbb{R}^n$, $|x|$ denotes the Euclidean norm, and $|x|_P^2$ is defined as $|x|_P^2 := x^\top P x$, where $P \in \mathbb{R}^{n \times n}$ is a symmetric matrix.   A C-set is defined as a convex, compact set containing the origin. Given a set $\mathbb{A} \subset \mathbb{R}^{n}$ and a matrix 
$K \in \mathbb{R}^{m\times n}$, we define
 $K \mathbb{A} := \{ K a \;|\; a \in \mathbb{A} \}.$ A function belongs to class $\mathcal{K}$ if it is continuous, zero at zero, and strictly increasing; a function belongs to class $\mathcal{K}_\infty$ if it is in class $\mathcal{K}$ and unbounded.

\section{Preliminaries}

\subsection{System Description}

Consider a non-affine, nonlinear, time-invariant, discrete-time system given by
\begin{equation}\label{eqsystem}
    x^{+} = f(x, u),
\end{equation}
where $x \in \mathbb{R}^n$ denotes the system state, $u \in \mathbb{R}^m$ the control input, and $x^{+}$ the successor state. 
The system is subject to state and input constraints
\begin{equation}
    x \in \mathbb{X}, \qquad  u \in \mathbb{U},
\end{equation}
where $\mathbb{X} \subset \mathbb{R}^n$ and $\mathbb{U} \subset \mathbb{R}^m$ are known C-sets, i.e., convex, compact sets containing the origin.

Assume that at each time step $k$, a goal-oriented command $u_{\text{ext}}(k) \in \mathbb{R}^m$ is available, which may be designed using either learning-based or data-driven methods but does not inherently guarantee safety or stability.  The objective of this work is to construct a safety--stability (S\textsuperscript{2})-set $\mathbb{S}(x) \subseteq \mathbb{U}$ such that the filtered control input
\begin{equation}\label{final}
    u  = \arg\min_{\zeta \in \mathbb{S}(x)} \, |u_{\text{ext}} - \zeta|^2
\end{equation}
guarantees both constraint satisfaction and asymptotic stability of the resulting closed-loop system, for all admissible initial states within the feasible operating region. The S\textsuperscript{2}-set $\mathbb{S}(x)$ is constructed implicitly through two cascaded OCPs, which are introduced sequentially in the following subsections.  To distinguish between them, we denote the control sequence by $\mathbf{v}$ and the predicted state by $z(i;x):= \phi(i; x, \mathbf{v})$ in the first OCP (the nominal MPC), and by $\mathbf{u}$ and $x(i;x):= \phi(i; x, \mathbf{u})$ in the second OCP (the filter). 
Here, $\phi(i; x, \cdot)$ denotes the $i$-step state transition from the initial state $x$ under the corresponding control sequence. 
By definition, $x(0;x) = z(0;x) = x$.




\subsection{Nominal MPC as a Copilot}

The first OCP is a nominal MPC problem, denoted by $\mathbb{P}_N(x)$, with prediction horizon $N$
\begin{equation}
    \mathbb{P}_N(x): \quad 
    V_N^*(x) := \min_{\mathbf{v}} 
    \bigl\{ V_N(x, \mathbf{v}) \ \big| \ \mathbf{v} \in \mathcal{V}_N(x) \bigr\},
\end{equation}
where the cost function $V_N(x, \mathbf{v})$ and the corresponding constraint set are defined as
\begin{equation}
\begin{aligned}
    &V_N(x, \mathbf{v}) := \sum_{i=0}^{N-1} \ell(z(i;x), v(i;x)) + V_f(z(N;x)), \\[2mm]
    &\mathcal{V}_N(x) := \Big\{ \mathbf{v} \ | \
    z(i;x) \in  \mathbb{X}, v(i;x)  \in \mathbb{U},  \forall i  \in  \mathbb{I}_{0:N-1}, z(N;x)  \in \mathbb{X}_f  \Big\}.
\end{aligned}
\end{equation}
Here, $\mathbb{X}_f$ denotes the terminal constraint set, while 
$\ell(\cdot,\cdot)$ and 
$V_f(\cdot)$  represent the stage and terminal cost functions, respectively.

Let $\mathcal{X}_N$ denote the set of initial states in $\mathbb{X}$ for which $\mathbb{P}_N(x)$ is feasible, i.e.,
\begin{equation}
    \mathcal{X}_N:= \bigl\{ x \in \mathbb{X} \,\big|\, \mathcal{V}_N(x) \neq \emptyset \bigr\}.
\end{equation}
For any $x \in \mathcal{X}_N$, the optimiser of $\mathbb{P}_N(x)$ is given by
\begin{equation}\label{vstar}
\begin{aligned}
    \mathbf{v}^*(x) 
    &= \arg\min_{\mathbf{v}} \bigl\{ V_N(x, \mathbf{v}) \ \big| \ \mathbf{v} \in \mathcal{V}_N(x) \bigr\}= \{ v^*(0; x), v^*(1; x), \ldots, v^*(N-1; x) \},
\end{aligned}
\end{equation}
and the corresponding optimal state sequence is
\begin{equation}\label{zstar}
    \mathbf{z}^*(x) = \{ z^*(0; x), z^*(1; x), \ldots, z^*(N; x) \}.
\end{equation}

To ensure recursive feasibility and asymptotic stability of the closed-loop system under the nominal MPC, the following standard assumptions are imposed. 
For further details, the reader is referred to~\cite[Chap.~2.4]{rawlings2020model}.

\begin{assumption} \label{assum:continuity}
The functions $f(\cdot,\cdot) : \mathbb{X} \times \mathbb{U} \rightarrow  \mathbb{X}$, 
$\ell(\cdot,\cdot) : \mathbb{X} \times \mathbb{U} \rightarrow \mathbb{R}_{\geq 0}$ and 
$V_f(\cdot) : \mathbb{X}_f \rightarrow \mathbb{R}_{\geq 0}$ are continuous, with
$f(0_{n\times 1}, 0_{m\times 1}) = 0$, $\ell(0_{n\times 1}, 0_{m\times 1}) = 0$ and $V_f(0_{n\times 1}) = 0$.
\end{assumption}


\begin{assumption} \label{assum:constraint_sets}
The sets $\mathbb{X}$,   $\mathbb{U}$, and   $\mathbb{X}_f \subseteq \mathbb{X}$ are  C-sets.
\end{assumption}

\begin{assumption}\label{assum:cost_bounds}
The stage cost \(\ell(\cdot,\cdot)\),  the terminal cost \(V_f(\cdot)\), and the terminal set  $\mathbb{X}_f$ satisfy the following properties:
\begin{enumerate}[label=(\alph*)] 
    \item For all $x\in \mathbb{X}_f$, there exists a $u\in\mathbb{U}$ satisfying
\[
\begin{aligned}
   & f(x,u)\in \mathbb{X}_f\\
   & V_f (f(x,u))-V_f(x)  \leq  - \ell(x,u).
\end{aligned}
\]
 \item There exists \(\mathcal{K}_\infty\)  functions \(\alpha_1(\cdot)\) and \(\alpha_2(\cdot)\) satisfying 
 \[
\begin{aligned}
   & \ell(x, u) \geq \alpha_1(|x|) \quad \forall x \in \mathcal{X}_N, u \in   \mathbb{U} \\
   & V_f (x) \leq  \alpha_2(|x|) \quad \forall x \in \mathbb{X}_f.
\end{aligned}
\]
\end{enumerate}
\end{assumption}



\begin{proposition}[{\cite[Proposition~2.16]{rawlings2020model}}]\label{lem:VN_upper}
Suppose that Assumptions~\ref{assum:continuity}, \ref{assum:constraint_sets}, 
and \ref{assum:cost_bounds} hold, and that the terminal set 
$\mathbb{X}_f$ contains the origin in its interior.  
Then, there exists a $\mathcal{K}_\infty$ function $\beta$ such that  
\[
    V_N^*(x) \;\le\; \beta(|x|), \qquad \forall\, x \in \mathcal{X}_N.
\]
\end{proposition}

\begin{remark}
To guarantee safety and stability, the stage cost function in MPC is typically required to be non-negative and continuous (see Assumption~\ref{assum:continuity}). While these properties enable rigorous analysis (such as recursive feasibility and asymptotic stability), they also imply that the nominal MPC cost may not directly encode the practical task that the system is meant to accomplish. In many applications, task-oriented objectives are inherently discontinuous or non-smooth. For instance, reinforcement-learning controllers commonly employ sparse or event-triggered reward structures, which introduce discontinuities; logic-based or mode-switching tasks lead to piecewise or hybrid cost functions; and human-generated teleoperation inputs may not correspond to any continuous cost at all. These characteristics make such task-driven objectives fundamentally incompatible with standard MPC cost requirements. This is precisely why the external command $u_{\mathrm{ext}}(k)$ plays a crucial role in the proposed framework: it can originate from any source, e.g., human intent, learning-based policies, heuristic planners, or discontinuous signals, and thus more faithfully captures the true task objective. The nominal MPC therefore acts merely as a safety-stability copilot, ensuring constraint satisfaction and stability while allowing $u_{\mathrm{ext}}(k)$ to dictate the high-level objective.
\end{remark}

\section{Main Results}
\subsection{Predictive Safety--Stability Filter}

With the optimal sequences $\mathbf{v}^*(x)$ and $\mathbf{z}^*(x)$ obtained from \eqref{vstar} and \eqref{zstar}, 
we can now introduce the proposed framework, the \textit{Predictive Safety--Stability Filter} (PS\textsuperscript{2}F).

Given any external command $u_{\text{ext}}$, the second OCP is formulated as
\begin{equation}\label{PS2F}
    \mathbb{P}_{f,M}(x): \quad 
    V_{f,M}^*(x) := \min_{\mathbf{u}} 
    \bigl\{ V_{f,M}(x, \mathbf{u}) \ \big| \ \mathbf{u} \in \mathcal{U}^a_M(x) \bigr\},
\end{equation}
where the cost function $V_{f,M}(x, \mathbf{u})$ and the corresponding constraint set $\mathcal{U}^a_M(x)$ are defined as
\begin{equation}
\begin{aligned}
&V_{f,M}(x,\mathbf{u}) := |u_\text{ext}-u(0;x)|^2 \\
&\begin{aligned}
    \mathcal{U}^a_M(x) := \Big\{ \mathbf{u} \ \Big| \ &
x(i;x) \in \mathbb{X}, ~ u(i;x) \in \mathbb{U}, ~ \forall i \in \mathbb{I}_{0:M-1}, \eqref{percons} ~\text{and}~\eqref{mequ}
\Big\}.
\end{aligned}
\end{aligned}
\end{equation}
Note that the external command $u_{\mathrm{ext}}$ is embedded directly into the cost  function of $\mathbb{P}_{f,M}(x)$. The performance constraint and terminal  constraint are given respectively by
\begin{subequations}
\begin{align}
    &L(\mathbf{x}_{0:M-1}(x), \mathbf{u}(x)) - L(\mathbf{z}^*_{0:M-1} (x), \mathbf{v}^*_{0:M-1}(x)) - a\ell(x(0;x), u(0;x)) \le 0, \label{percons}\\
    &x(M;x) = z^*(M; x), \label{mequ}
\end{align}
\end{subequations}
where $a\ge 0$ is a tunable parameter and $M\in \mathbb{I}_{1:N}$ is the prediction
horizon of the second OCP.
The function $L(\cdot,\cdot)$ represents the cumulative stage cost over the first $M$ steps, defined as
\begin{equation}
    L(\mathbf{x}_{0:M-1}(x), \mathbf{u}(x)) := \sum_{i=0}^{M-1} \ell\bigl(x(i;x), u(i;x)\bigr).
\end{equation}
Similarly, we have 
\[
L(\mathbf{z}^*_{0:M-1}(x), \mathbf{v}^*_{0:M-1}(x)) = \sum_{i=0}^{M-1} \ell\bigl(z^*(i;x),v^*(i;x)\bigr).\]
The optimiser of $\mathbb{P}_{f,M}(x)$ is denoted by
\begin{equation}\label{ustar}
\begin{aligned}
    \mathbf{u}^*(x) 
    &= \arg\min_{\mathbf{u}} \bigl\{ V_{f,M}(x, \mathbf{u}) \ \big| \ \mathbf{u} \in \mathcal{U}^a_M(x) \bigr\} = \{ u^*(0; x), u^*(1; x), \ldots, u^*(M-1; x) \},
\end{aligned}
\end{equation}
and the corresponding optimal state sequence is
\begin{equation}\label{xstar}
    \mathbf{x}^*(x) = \{ x^*(0; x), x^*(1; x), \ldots, x^*(M; x) \}.
\end{equation}
Finally, the filtered control law is defined as
\begin{equation}\label{equ_finalcontrol}
    u = \kappa(x, u_{\text{ext}}) := u^*(0; x).
\end{equation}
Since only the first control component influences the cost function, 
\eqref{equ_finalcontrol} is equivalent to \eqref{final} through the definition of the 
S\textsuperscript{2}-set $\mathbb{S}(x)$, given by
\begin{equation} \label{eq:admissible_controls_sim}
    \mathbb{S}(x)
    := \{\, u(0;x)\;|\; \mathbf{u} \in \mathcal{U}^a_M(x)\,\}.
\end{equation}
The complete control procedure is presented in Algorithm~\ref{alg:ps2f_simple}, while Fig.~\ref{fig:mpc-indicator} illustrates the overall structure of the proposed PS\textsuperscript{2}F framework.

\begin{figure}[!tb]
    \centering
    \includegraphics[width=0.4\linewidth]{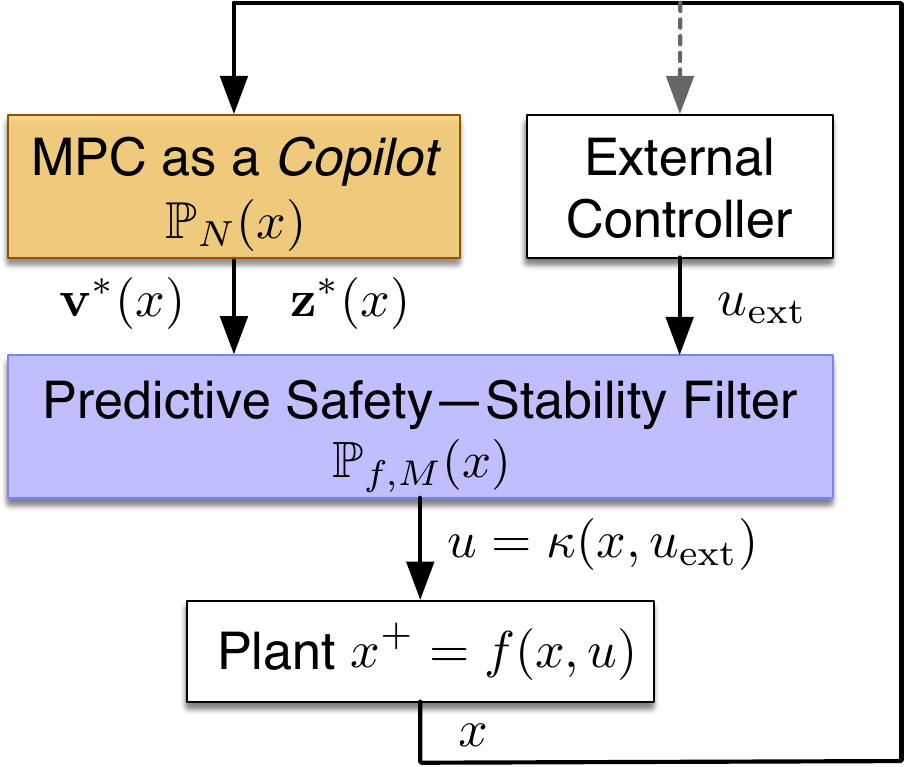}
   \caption{Schematic of the proposed PS\textsuperscript{2}F framework. 
The nominal MPC serves as a \textit{copilot}, generating safety- and stability-certified trajectories. 
The filtering OCP $\mathbb{P}_{f,M}(\cdot)$ then refines the external command within a safe and stable region before applying the final control action to the system plant. The dashed arrow above the external controller indicates that the external command may be independent of the system state, i.e., operating in open-loop mode.}

    \label{fig:mpc-indicator}
\end{figure}

\begin{algorithm}
\caption{\textbf{PS\textsuperscript{2}F}}
\label{alg:ps2f_simple}
\begin{algorithmic}[1]
\State \textbf{Offline:}
Specify prediction horizons $N$  and $M\!\le\!N$;
cost functions $\ell(\cdot,\cdot)$ and $V_f(\cdot)$; constraint sets $\mathbb{X},\mathbb{U},\mathbb{X}_f$; 
and parameter  $a\in[0,1)$.
\State \textbf{Initialisation:} Set time $k\leftarrow0$ and measure $x(0)$.

  \State \textbf{Step One-Nominal MPC:}
  Solve $\mathbb{P}_N\big(x(k)\big)$ to obtain $\mathbf{v}^*(x(k))$ and $\mathbf{z}^*(x(k))$. 
 
  \State \textbf{Step Two-Receive external command:}
  Obtain $u_{\mathrm{ext}}(k)$ from any goal-oriented controller.

  \State \textbf{Step Three-PS\textsuperscript{2}F:}
  Solve $\mathbb{P}_{f,M}\big(x(k)\big)$  to obtain 
  $\mathbf{u}^*(x(k))$ and $\mathbf{x}^*(x(k))$. 
  Set the filtered control input
  \[
  u(k)\leftarrow u^*(0;x(k)) \quad \text{(cf.\ \eqref{equ_finalcontrol}).}
  \]

  \State \textbf{Step Four-Apply and propagate:}
  Apply $u(k)$ to the system \eqref{eqsystem} to obtain $x(k{+}1)=f\big(x(k),u(k)\big)$.
  \State $k \leftarrow k + 1$
and go to \textbf{Step One}. 
 
\end{algorithmic}
\end{algorithm}

Before analysing the closed-loop behaviour, we first establish several fundamental properties of the second OCP, $\mathbb{P}_{f,M}(x)$, including its feasibility and the role of the tunable parameters $a$ and $M$. 
Unless otherwise specified, Assumptions~\ref{assum:continuity}, \ref{assum:constraint_sets}, and \ref{assum:cost_bounds} are assumed to hold throughout this subsection. 
For readability, the proofs of the following results are provided in Appendix~\ref{proof_propos}.

\begin{proposition}\label{lem_fea}
 $\mathbb{P}_{f,M}(x)$ is feasible for all $x \in \mathcal{X}_N$; that is, $\mathcal{U}^a_M(x) \neq \emptyset$ for every $x \in \mathcal{X}_N$.
\end{proposition}

\begin{proposition}\label{a0lemma}
Suppose that  $x \in \mathcal{X}_N$ and $a = 0$.
\begin{itemize}
    \item[(a)] The below equality holds:
    $$L(\mathbf{x}^*_{0:M-1}(x), \mathbf{u}^*(x)) = L(\mathbf{z}^*_{0:M-1}(x), \mathbf{v}^*_{0:M-1}(x)).$$ 
    \item[(b)] Moreover, if the OCP $\mathbb{P}_N(x)$ is strictly convex, then the filtered control law reduces to the nominal MPC action. In particular,
    \[
        u = u^*(0;x) = v^*(0;x),
    \]
    and the S\textsuperscript{2}-set degenerates to a singleton:
    \[
        \mathbb{S}(x) = \{ v^*(0;x) \}.
    \]
\end{itemize}
\end{proposition}

\begin{proposition}\label{alemma}
Suppose that $x \in \mathcal{X}_N$ and $a_1 \ge a_2 \ge 0$.
Then, the corresponding S\textsuperscript{2}-sets satisfy the nested inclusion
\[
    \mathbb{S}^{a_2}(x) \subseteq \mathbb{S}^{a_1}(x),
\]
where $\mathbb{S}^{a_i}(x)$ denotes the S\textsuperscript{2}-set associated with the parameter $a = a_i$.
\end{proposition}

\begin{proposition}\label{M1lemma}
Suppose that $x \in \mathcal{X}_N$ and $M = 1$.  
Assume that the mapping $u \mapsto f(x,u)$ is injective on $\mathbb{U}$.  
Then, the filtered control law reduces to the nominal MPC action.  
In particular,
\[
    u(0;x) = u^*(0;x) = v^*(0;x),
\]
and the S\textsuperscript{2}-set degenerates to a singleton:
\[
    \mathbb{S}(x) = \{ v^*(0;x)\}.
\]
\end{proposition}

\begin{proposition}\label{mlemma}
Suppose that $x \in \mathcal{X}_N$ and $i\ge j$, where $i,j\in\mathbb{I}_{1:N}$.
Then, the corresponding S\textsuperscript{2}-sets satisfy the nested inclusion
\[
\mathbb{S}_j(x) \subseteq \mathbb{S}_i(x),
\]
where $\mathbb{S}_i(x)$ denotes the S\textsuperscript{2}-set associated with prediction horizon  $M = i$.
\end{proposition}

\begin{remark}
Proposition~\ref{lem_fea} shows that the filter problem $\mathbb{P}_{f,M}(x)$ does not require any additional assumptions beyond those of the nominal MPC problem $\mathbb{P}_N(x)$. 
As long as $x \in \mathcal{X}_N$, the filter remains feasible. 
Therefore, classical techniques for enlarging the region of attraction $\mathcal{X}_N$ (e.g.,~\cite{limon2005enlarging}) can be directly employed to extend the safe and stable operating region.
\end{remark}

\begin{remark}
Propositions~\ref{a0lemma}, \ref{alemma}, \ref{M1lemma}, and~\ref{mlemma} jointly characterise how the design parameters $a$ and $M$ determine the geometry of the S\textsuperscript{2}-set.  
Proposition~\ref{a0lemma} shows that when $a=0$, the S\textsuperscript{2}-set collapses to the nominal MPC action.  
Proposition~\ref{alemma} establishes that the S\textsuperscript{2}-set is monotone in $a$, with larger values producing larger sets.  
Similarly, Propositions~\ref{M1lemma} and~\ref{mlemma} show that the prediction horizon $M$ also induces a monotone expansion of the S\textsuperscript{2}-set, and that for $M=1$ the filtered input coincides with the nominal MPC control.  
Taken together, these results demonstrate that increasing either $a$ or $M$ reduces conservatism by enlarging the S\textsuperscript{2}-set, whereas smaller values make the filter behave more like the nominal MPC law.  
This trend is further illustrated in the numerical examples presented later.
\end{remark}

To more clearly illustrate the geometry of the S\textsuperscript{2}-set $\mathbb{S}(x)$, we consider a representative example based on a general linear
system. Although it was previously defined only implicitly through the two OCPs, in this example, we explicitly derive the analytical description of $\mathbb{S}(x)$, making its structure transparent.

\begin{proposition}\label{linearsys}
Suppose that $f(x,u) = A x + B u$, where 
$A \in \mathbb{R}^{n\times n}$ and 
$B \in \mathbb{R}^{n\times m}$ are system matrices, and $(A,B)$ is controllable. 
The nominal MPC is designed with the quadratic stage and terminal costs, given by
\[
    \ell(x,u) = x^\top Q x + u^\top R u,
    \qquad
    V_f(x) = x^\top P x,
\]
where $Q > 0$, $R > 0$, and $P>0$ is the positive definite solution of the discrete-time algebraic Riccati equation associated with $(A,B,Q,R)$. 
We further assume that the safety constraints in both OCPs, as well as the terminal constraint of the first OCP, are inactive. Then, the admissible input set $\mathcal{U}^a_M(x)$ is the intersection of an affine subspace and a quadratic set, i.e., an ellipsoid sliced by a linear constraint:
\[
\begin{aligned}
  \mathcal{U}^a_M(x)
    =
    \Bigl\{
        \mathbf{u} \in \mathbb{R}^{Mm}
        \,\Big|\,
        &A_{\mathrm{eq}} \mathbf{u} = b_{\mathrm{eq}}x,\mathbf{u}^\top H\,\mathbf{u}
        + 2 x^\top F\,\mathbf{u}
        + x^\top G x \le0
    \Bigr\},   
\end{aligned}
\]
and hence,
\[
 \mathbb{S}(x) =e_1\mathcal{U}^a_M(x),
\]
where $A_{\mathrm{eq}}$, $b_{\mathrm{eq}}$, $H$, $F$, $G$, and $e_1$ are constant matrices 
whose explicit expressions are provided in Appendix A.  
\end{proposition}

\subsection{Theoretical Analysis}

In this subsection, we analyse the closed-loop properties of the proposed PS\textsuperscript{2}F framework. 
We begin with the nominal case, where all parameters are held constant, and then extend the analysis to the mode-scheduling case, in which the parameters vary over time. 
For completeness, the proofs of all theoretical results presented in this subsection are provided in Appendix~\ref{proof_theors}.

\begin{theorem}\label{thm:Recursive feasibility and stability}
Suppose that Assumptions~\ref{assum:continuity}, \ref{assum:constraint_sets}, and 
\ref{assum:cost_bounds} hold, and that the terminal set $\mathbb{X}_f$ contains 
the origin in its interior. Let the initial state satisfy $x(0) \in \mathcal{X}_N$. 
Then, for any external command $u_{\mathrm{ext}}(k)$, if $a \in [0,1)$, the  system~\eqref{eqsystem} under the control law~\eqref{equ_finalcontrol} 
satisfies the following properties:
\begin{enumerate}[label=(\alph*)]
    \item 
    The safety constraints are satisfied for all $k\in\mathbb{I}_{\ge 0}$.
    \item The closed-loop system is asymptotically stable.
\end{enumerate}
\end{theorem}

Based on Theorem~\ref{thm:Recursive feasibility and stability}, we now extend the proposed framework to the time-varying case, in which the design parameters $a(k)$ and $M(k)$ evolve over time (see Algorithm~\ref{alg:ps2f_complex}). 
This extension enables the controller to adapt its behaviour dynamically, i.e., prioritising safety during exploration and progressively enforcing stability during exploitation. The following theorem formalises these properties.

\begin{theorem}\label{thm:timevarying}
Suppose that Assumptions~\ref{assum:continuity}, \ref{assum:constraint_sets}, and 
\ref{assum:cost_bounds} hold, and that the terminal set $\mathbb{X}_f$ contains 
the origin in its interior. Let the initial state satisfy $x(0) \in \mathcal{X}_N$.   Then, for any given external command $u_{\text{ext}}(k)$, if $M(k)\in\mathbb{I}_{1:N}$ and $a(k)\in[0,\infty)$, the 
system~\eqref{eqsystem} under the control law~\eqref{equ_finalcontrol} 
satisfies the following properties:
\begin{enumerate}[label=(\alph*)] 
    \item The safety constraints are satisfied for all $k\in\mathbb{I}_{\ge 0}$.
    \item Furthermore, if there exists a time instant $K_s\in\mathbb{I}_{\ge 0}$  such that $\sup_{k\ge K_s} a(k) < 1 $, then the closed-loop system is asymptotically stable for all $k \in\mathbb{I}_{\ge K_s}$.
\end{enumerate}
\end{theorem}

\begin{algorithm}
\caption{\textbf{PS\textsuperscript{2}F with Mode Scheduling}}
\label{alg:ps2f_complex}
\begin{algorithmic}[1]
\State \textbf{Offline:}
Specify prediction horizons $N$  and $M(k)\in\mathbb{I}_{1:N}$;
cost functions $\ell(\cdot,\cdot)$ and $V_f(\cdot)$; constraint sets $\mathbb{X},\mathbb{U},\mathbb{X}_f$; 
and mode-scheduling parameter $a(k)\!\ge\!0$.
\State \textbf{Initialisation:} Set time $k\leftarrow0$ and measure $x(0)$.

  \State \textbf{Step One-Nominal MPC:}
  Solve $\mathbb{P}_N\big(x(k)\big)$ to obtain $\mathbf{v}^*(x(k))$ and $\mathbf{z}^*(x(k))$. 
 
  \State \textbf{Step Two-Receive external command:}
  Obtain $u_{\mathrm{ext}}(k)$ from any goal-oriented controller.

  \State \textbf{Step Three-Mode scheduling:}
  Adjust $a(k)$ according to the desired operation mode:
  \If{safe exploration phase}
    \State Choose a large $a(k) \ge 1$ to prioritise constraint satisfaction (safety-first mode).
  \ElsIf{stable exploitation phase}
    \State Choose a small $a(k) \le 1 - \bar{a}$ (with $\bar{a}\in(0,1]$) to enhance asymptotic stability.
  \EndIf

  \State \textbf{Step Four-PS\textsuperscript{2}F:} Receive the scheduled parameters $a(k)$ and $M(k)$.
  Solve $\mathbb{P}_{f,M}\big(x(k)\big)$  to obtain 
  $\mathbf{u}^*(x(k))$ and $\mathbf{x}^*(x(k))$. 
  Set the filtered control input
  \[
  u(k)\leftarrow u^*(0;x(k)) \quad \text{(cf.\ \eqref{equ_finalcontrol}).}
  \]

  \State \textbf{Step Five-Apply and propagate:}
  Apply $u(k)$ to the system \eqref{eqsystem} to obtain $x(k{+}1)=f\big(x(k),u(k)\big)$.
  \State $k \leftarrow k + 1$
and go to \textbf{Step One}. 
 
\end{algorithmic}
\end{algorithm}

\begin{remark}
Theorem~\ref{thm:timevarying} provides the theoretical foundation for resolving the classical exploration–exploitation dilemma within a unified predictive control framework. Property~(a) guarantees that the system remains safe even when stability is intentionally relaxed (e.g., by choosing $a(k) \ge 1$), thereby enabling the controller to explore a broader range of control actions or state trajectories without violating safety constraints—this corresponds to the \textit{safe exploration mode}. 
In contrast, Property~(b) ensures that, after any finite exploration phase, the system trajectories asymptotically converge once the controller switches to the \textit{exploitation mode} (by persistently enforcing $a(k) \le 1 - \bar{a}$, with $\bar{a}\in(0,1]$), where stability dominates.

Different from the parameter $a(k)$, which directly determines the stability 
behaviour of the closed-loop system, the horizon parameter $M(k)$ does not 
explicitly influence stability. Nevertheless, $M(k)$ plays a complementary and 
practically important role: increasing $M(k)$ enlarges the 
S\textsuperscript{2}-set, thereby reducing conservatism and allowing the controller 
to preserve more of the external command. This, 
however, comes at the price of higher computational complexity. Conversely, smaller 
values of $M(k)$ shrink the S\textsuperscript{2}-set, making the filter more 
restrictive but computationally cheaper. Adjusting both $a(k)$ and $M(k)$ thus provides a principled 
mechanism for shaping how the system transitions from exploration to exploitation. This capability to safely transition between exploration and exploitation 
phases within a single, theoretically grounded framework represents a central 
contribution of the proposed PS\textsuperscript{2}F approach. This feature will be 
further illustrated in the robotic navigation task in the following section. 
\end{remark}

\begin{remark}
   In situations where the terminal set $\mathbb{X}_f$ does not have an interior, for example when $\mathbb{X}_f=\{0_{n\times 1}\}$, the upper bound in
Proposition~\ref{lem:VN_upper} cannot be established, and thus the asymptotic 
stability in the Lyapunov sense cannot be concluded directly~\cite[Chap.~2.4]{rawlings2020model}. 
Nevertheless, convergence of the state to the origin can still be guaranteed. Since $V_N^*(x(k))$ is nonincreasing along closed-loop trajectories 
(see~\eqref{vdecrease} or~\eqref{vdecrease2}) and satisfies 
$V_N^*(x(k)) \ge 0$ for all $k$, the Monotone Convergence Theorem implies that 
$V_N^*(x(k))$ converges to a finite limit as $k\to\infty$. 
Moreover, the decrease condition forces this limit to be zero, which in turn 
implies that $x(k)$ converges to the origin.
\end{remark}

\section{Simulation}
 
In this section, we validate the proposed PS\textsuperscript{2}F framework through three simulation studies.  Each case study highlights a different aspect of the method: the evolution of the S\textsuperscript{2}-set over time,  the influence of design parameters on conservatism, and the integration of PS\textsuperscript{2}F with a 
goal-oriented external controller in robotic navigation. These simulations demonstrate how the proposed filter consistently enforces safety and stability 
while allowing flexible, high-level control behaviours.

\subsection{Case Study 1: Evolution Over Time}

Consider the linear system $x^+=Ax+Bu$, with
\[
A = 
\begin{bmatrix}
1 & 1\\
0 & 1
\end{bmatrix}, \qquad
B = 
\begin{bmatrix}
1 & 0\\
0 & 1
\end{bmatrix},
\]
subject to the state and input constraints
\[
\mathbb{X} = [-2,2]^2, \qquad 
\mathbb{U} = [-1,1]^2.
\]
The initial state is chosen as $x(0) = [2,-2]^\top$.

The MPC stage cost is defined as $\ell(x,u) = x^\top Q x + u^\top R u$ with 
$Q = 10 I_2$, $R = I_2$, and prediction horizon $N = 5$. 
The terminal cost and terminal set are chosen as
\[
    V_f(x) = x^\top P x,
    \qquad
    \mathbb{X}_f = \{\, x \in \mathbb{R}^2 \mid x^\top P x \le \gamma \,\},
\]
where
\[
    P =
    \begin{bmatrix}
        10.92 & 0.92 \\
        0.92  & 11.85
    \end{bmatrix}
\]
is the positive definite solution of the discrete-time algebraic Riccati equation 
associated with $(A,B,Q,R)$, and $\gamma = 7.28$ is selected so that 
$\mathbb{X}_f$ is the largest ellipsoid contained within the state and input 
constraints and invariant under the terminal controller $u = -K_{\mathrm{LQR}} x$. 
It is straightforward to verify that this MPC setup satisfies 
Assumptions~\ref{assum:continuity}, \ref{assum:constraint_sets}, and 
\ref{assum:cost_bounds}. 
The PS\textsuperscript{2}F parameters are chosen as $M = 2$ and $\alpha = 0.95$.

To make the scenario more challenging, we choose the external command as
\[
u_{\mathrm{ext}}(k) = 
\begin{bmatrix}
u_{1,\mathrm{ext}}(k)\\
u_{2,\mathrm{ext}}(k)
\end{bmatrix}
=
\begin{bmatrix}
-1.2\cos(0.2k + 0.2)\\
0.1\,x_2(k)
\end{bmatrix}.
\]
The first component $u_{1,\mathrm{ext}}(k)$ is an open-loop signal that may violate the input constraints, while the second component $u_{2,\mathrm{ext}}(k)$ corresponds to an unstable feedback law. Indeed, applying $u_{2,\mathrm{ext}}(k)$ directly to the system leads to
$x_2(k+1) = 1.1\,x_2(k)$, 
which diverges over time and is therefore intrinsically unstable.

The results are shown in Figs.~\ref{fig:ps2f_3d} and~\ref{fig:ps2f_evolution_boxed}.
Fig.~\ref{fig:ps2f_3d} provides a macroscopic, time-evolving visualisation of the 
PS\textsuperscript{2}F mechanism. It illustrates how the S\textsuperscript{2}-set 
$\mathbb{S}(x(k))$ evolves as the system state changes, forming a grey surface in the 
3D space spanned by $(u_1, u_2, k)$. Superimposed on this surface are the trajectories 
of the external command $u_{\mathrm{ext}}(k)$ and the filtered control $u(k)$, which 
clearly show how unsafe or unstable external actions are projected onto the safe–stable 
region. Fig.~\ref{fig:ps2f_evolution_boxed} complements this with a microscopic, 
step-by-step perspective. Each panel shows a 2D cross-section of the 
S\textsuperscript{2}-set at a given time $k$, together with the external command, the 
filtered control, and the direction of correction applied by the filter. These 
snapshots also align with Proposition~\ref{linearsys}: as $x(k)$ approaches the origin, the corresponding S\textsuperscript{2}-set $\mathbb{S}(x(k))$ progressively shrinks, eventually collapsing to the singleton $\{0_{m\times 1}\}$. Overall, 
these visualisations demonstrate how PS\textsuperscript{2}F systematically modifies the 
external command at each time step and how the geometry of $\mathbb{S}(x(k))$ governs 
the filter’s corrective behaviour.

\begin{figure}[t]
    \centering
    \includegraphics[width=0.6\linewidth]{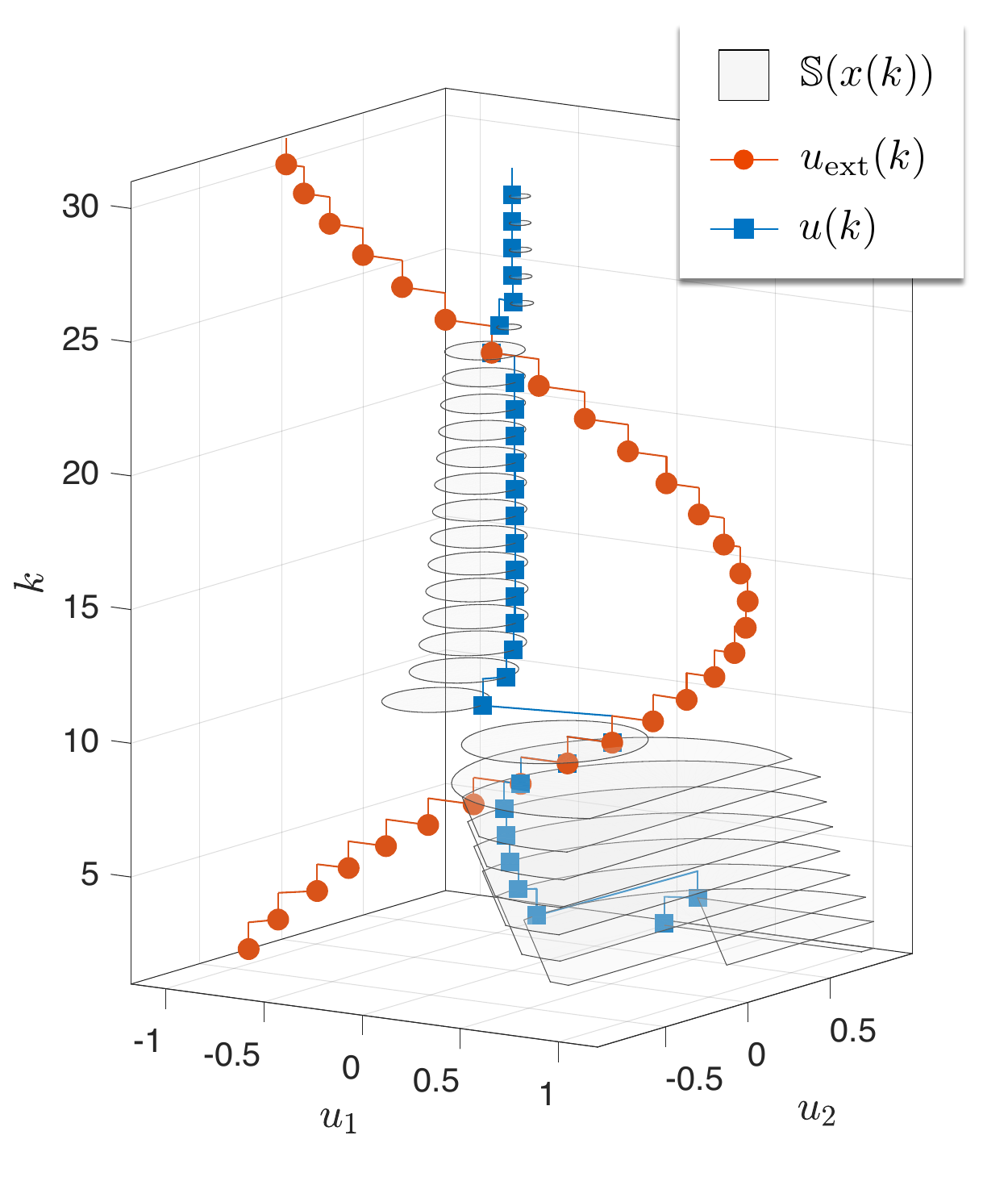}
\caption{Three-dimensional illustration of the PS\textsuperscript{2}F operation over time. 
The S\textsuperscript{2}-set $\mathbb{S}(x(k))$ (grey surface) evolves with the system state, 
while the external command  $u_{\mathrm{ext}}(k)$ (red circle) and the filtered control $u(k)$ (blue square) 
are visualised along the time axis.}

    \label{fig:ps2f_3d}
\end{figure}

\begin{figure*}[t]
  \centering
  \setlength{\fboxsep}{2pt}   
  \setlength{\fboxrule}{0.5pt}  
  \fbox{%
  \begin{minipage}{1\textwidth}
    \centering
    \setlength{\tabcolsep}{-4pt} 

    \begin{tabular}{ccccc}
      \includegraphics[width=0.21\linewidth]{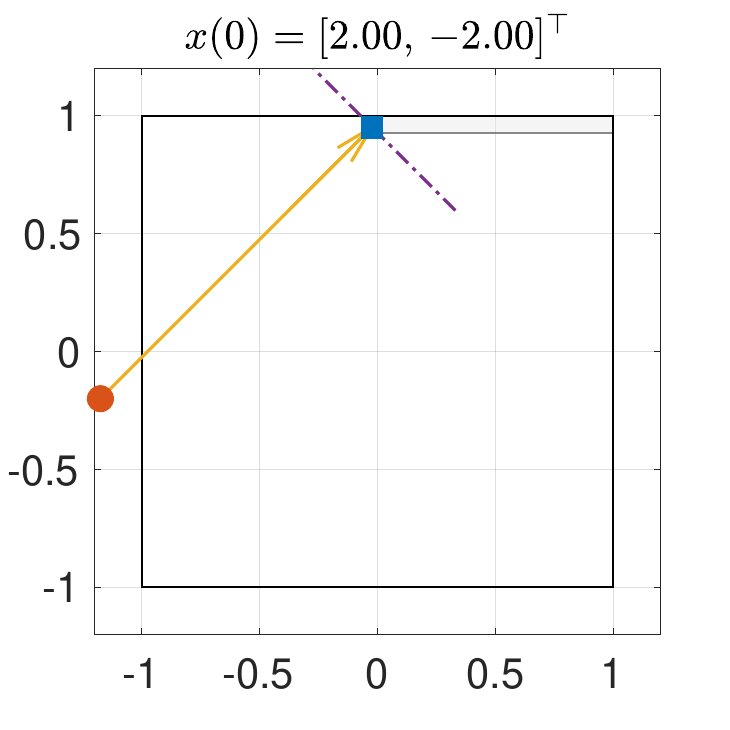} &
      \includegraphics[width=0.21\linewidth]{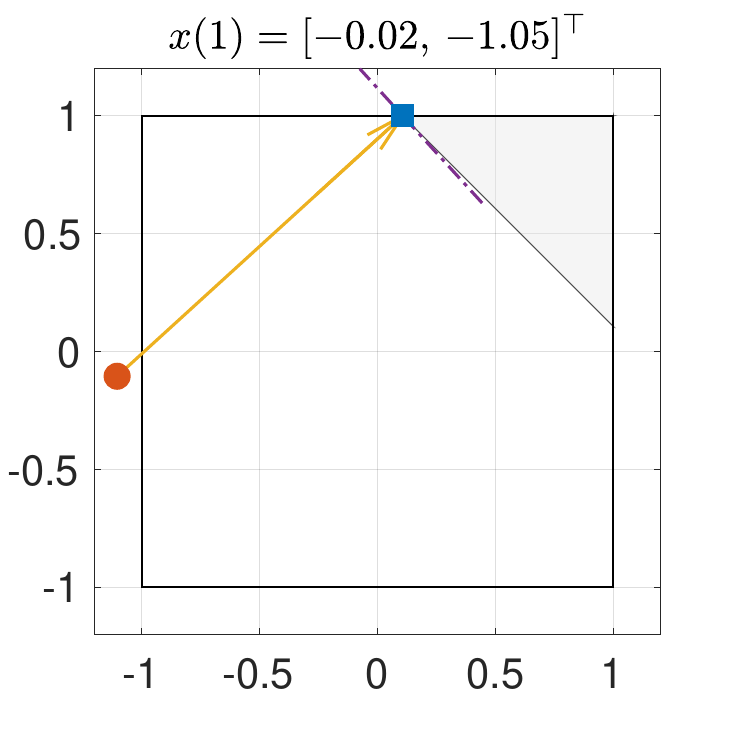} &
      \includegraphics[width=0.21\linewidth]{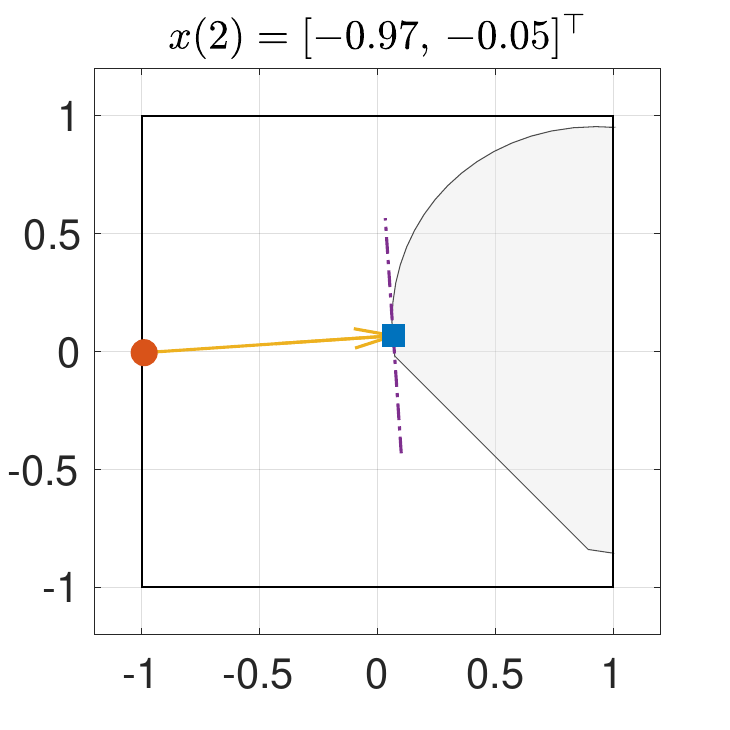} &
      \includegraphics[width=0.21\linewidth]{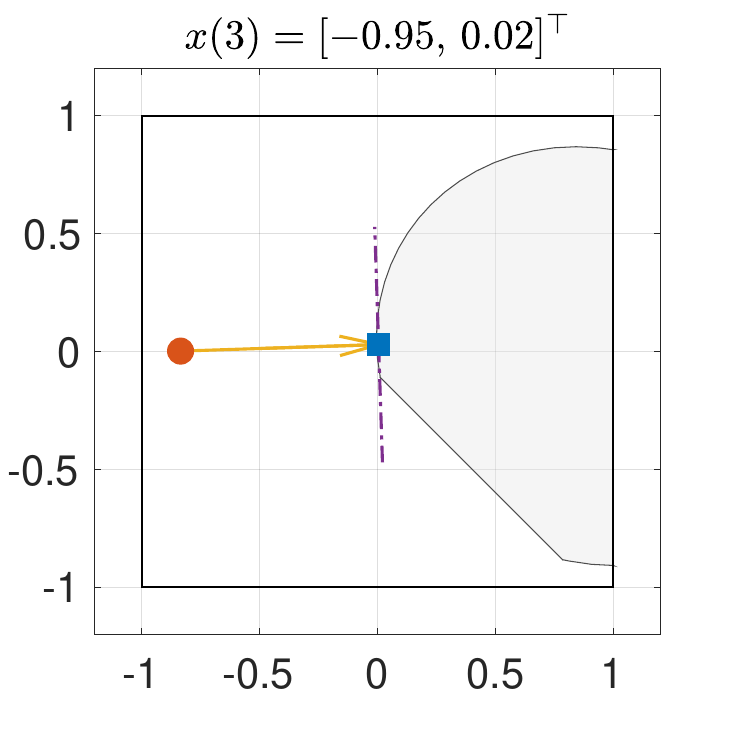} &
      \includegraphics[width=0.21\linewidth]{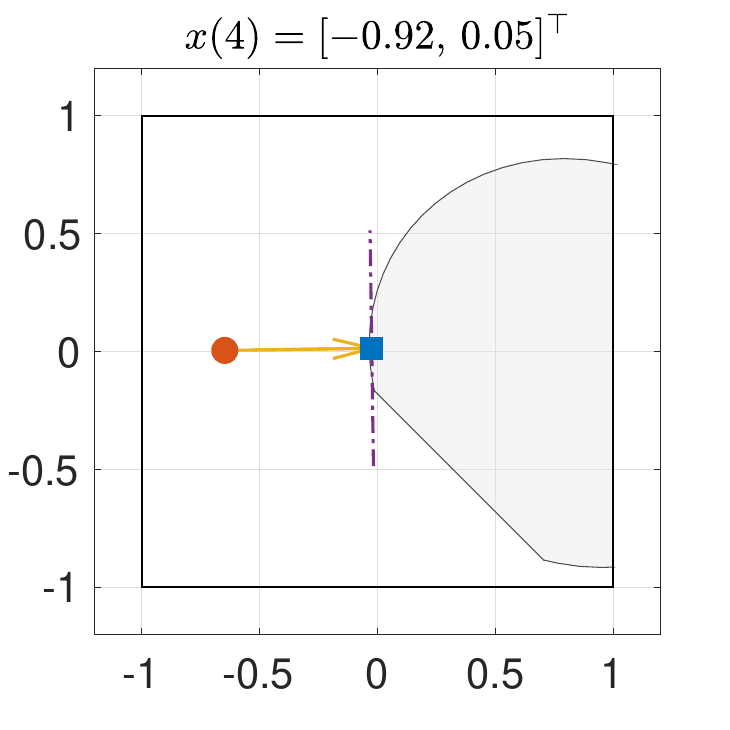} \\[-4pt]
      (a) $k=0$ & (b) $k=1$ & (c) $k=2$ & (d) $k=3$ & (e) $k=4$ \\
      \includegraphics[width=0.21\linewidth]{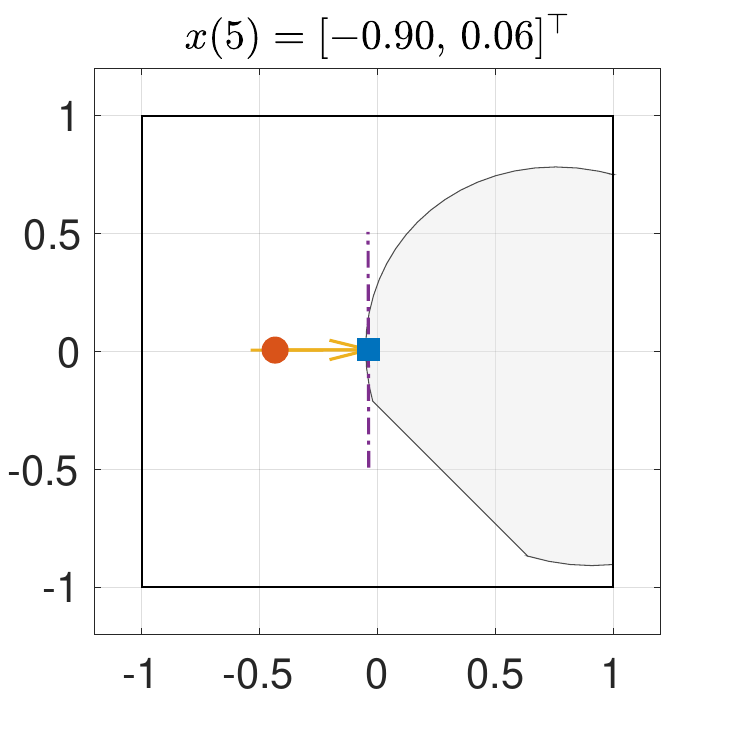} &
      \includegraphics[width=0.21\linewidth]{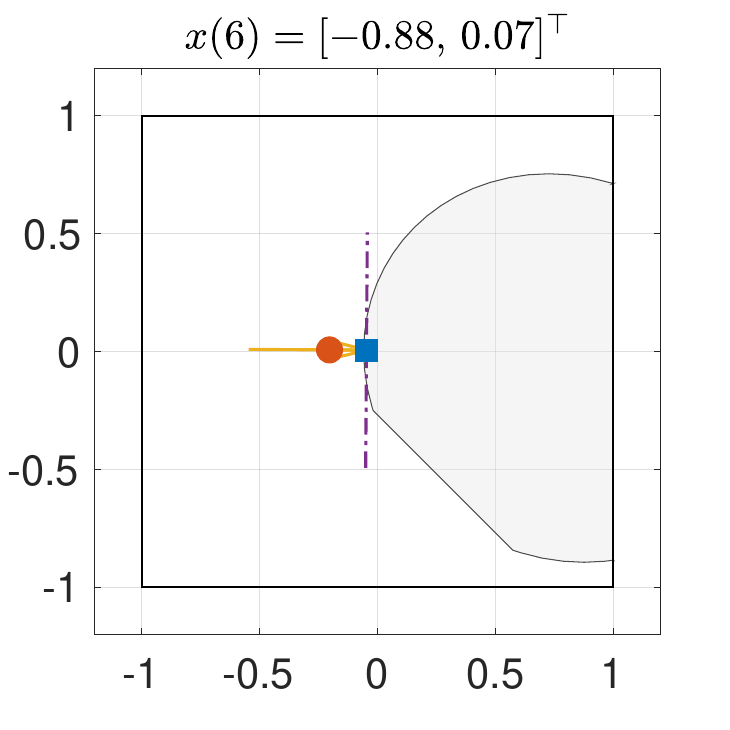} &
      \includegraphics[width=0.21\linewidth]{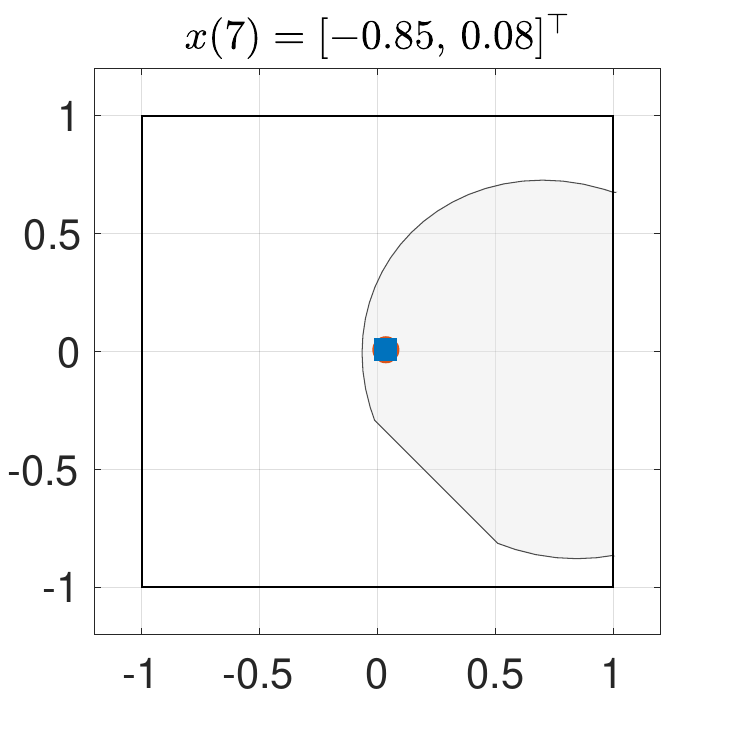} &
      \includegraphics[width=0.21\linewidth]{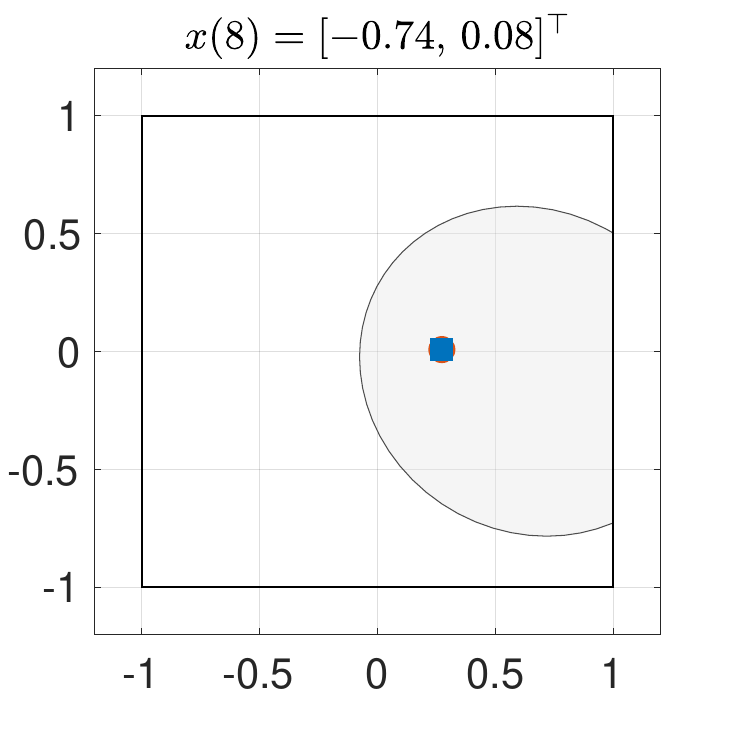} &
      \includegraphics[width=0.21\linewidth]{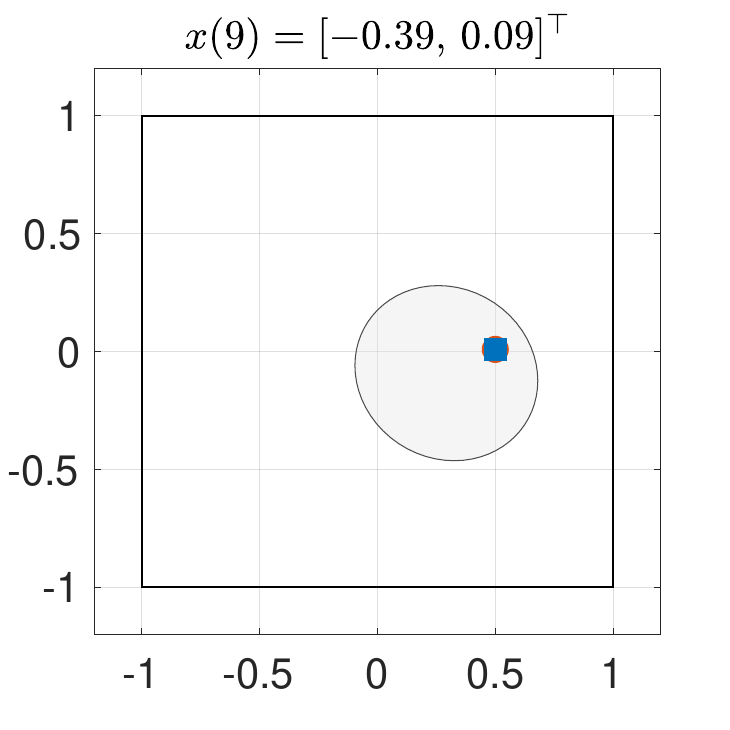} \\[-4pt]
      (f) $k=5$ & (g) $k=6$ & (h) $k=7$ & (i) $k=8$ & (j) $k=9$ \\
      \includegraphics[width=0.21\linewidth]{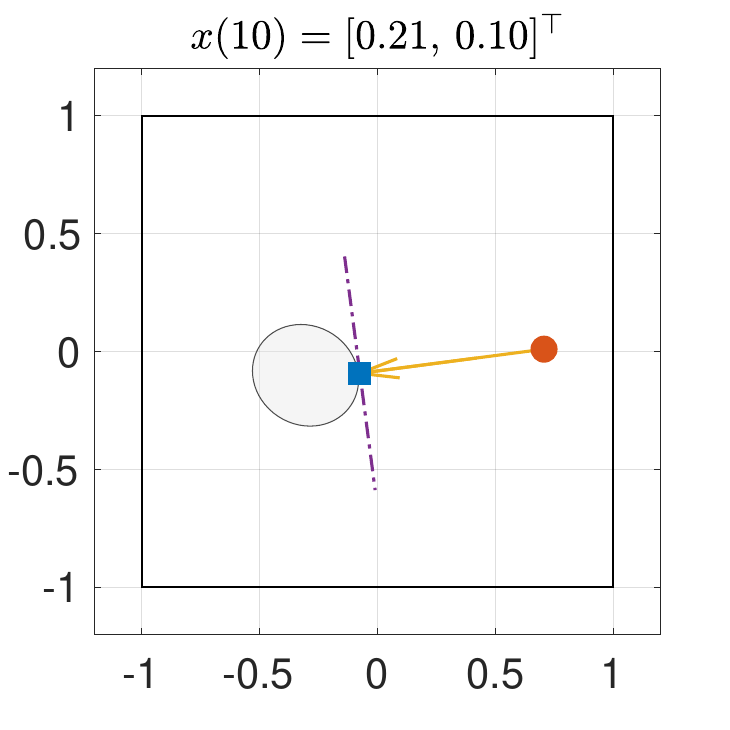} &
      \includegraphics[width=0.21\linewidth]{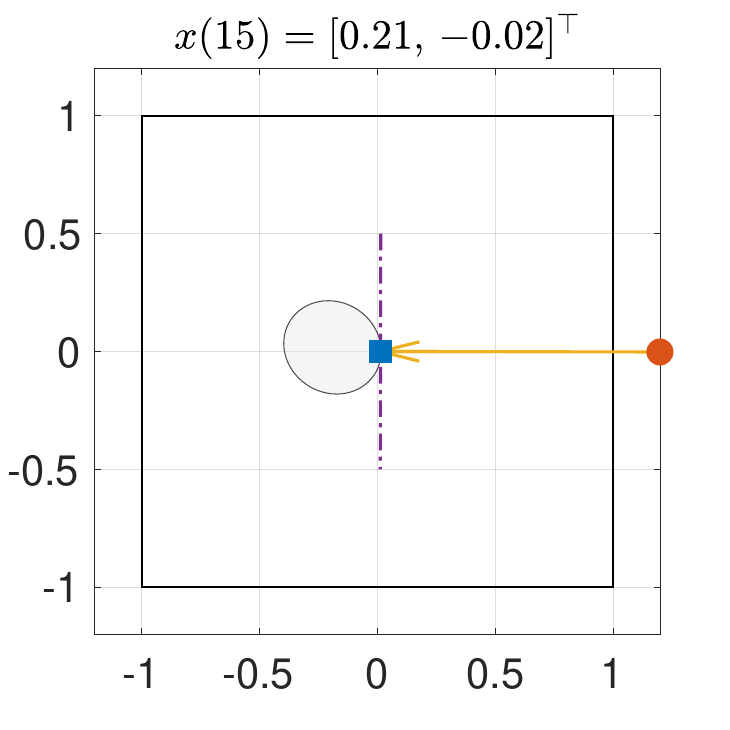} &
      \includegraphics[width=0.21\linewidth]{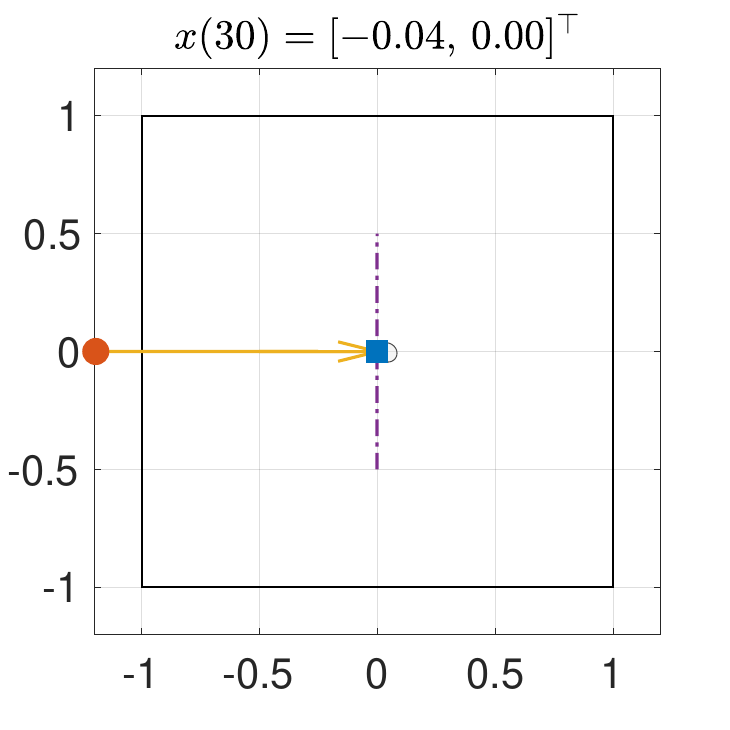} &
      \includegraphics[width=0.21\linewidth]{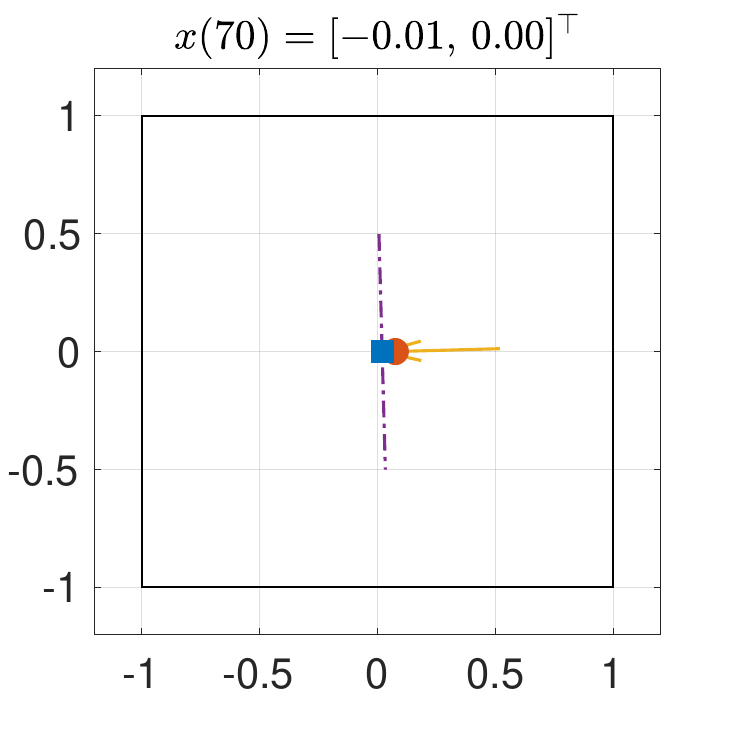} &
      \includegraphics[width=0.21\linewidth]{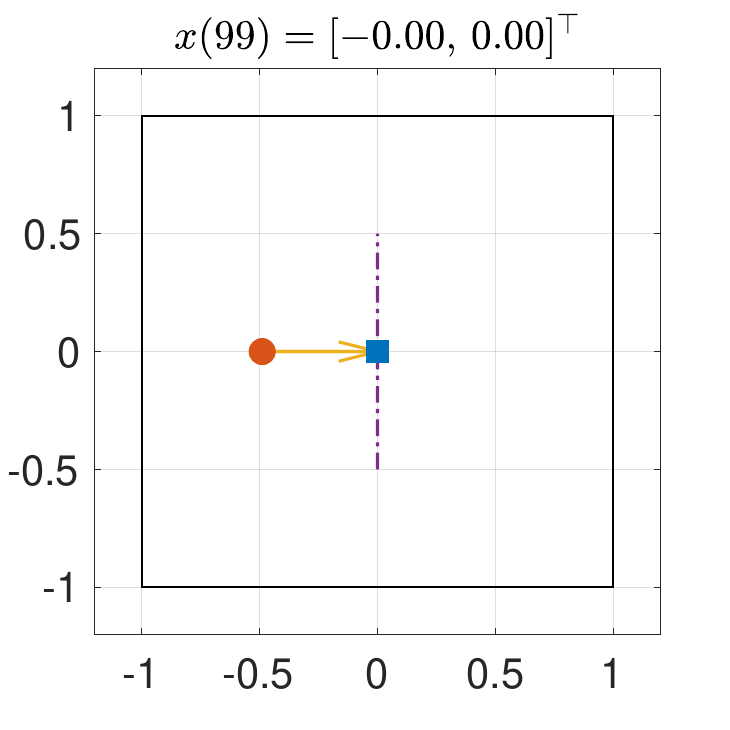} \\[-4pt]
      (k) $k=10$ & (l) $k=15$ & (m) $k=30$ & (n) $k=70$ & (o) $k=99$
    \end{tabular}

\caption{Evolution of the PS\textsuperscript{2}F filtering process along the $u_1$--$u_2$ axes.
Each panel shows the S\textsuperscript{2}-set $\mathbb{S}(x(k))$ (grey surface), 
the external command $u_{\mathrm{ext}}(k)$ (red circle), 
and the filtered control $u(k)$ (blue square). 
The orange arrow depicts the corrective action generated by PS\textsuperscript{2}F, 
which corresponds to the perpendicular projection of $u_{\mathrm{ext}}(k)$ onto the admissible set. 
This projection direction is orthogonal to the local tangent of the boundary 
$\partial\mathbb{S}(x(k))$ (purple dash-dotted line), ensuring that the filtered input 
remains within the safety--stability domain while staying as close as possible to the 
external command.}

    \label{fig:ps2f_evolution_boxed}
  \end{minipage}
  } 
\end{figure*}

\subsection{Case Study 2: Evolution Over Parameters}

To investigate how the design parameters influence the geometry of the S\textsuperscript{2}-set and thus the level of conservatism of the PS\textsuperscript{2}F, we systematically vary the key parameters and plot the resulting S\textsuperscript{2}-sets at the same state. For consistency, this parametric study is conducted on the same linear system used in Case Study~1, and we only focus on the case where the state is in its initial condition $x = x(0) = [2,-2]^\top$.

We begin with the parameters of the PS\textsuperscript{2}F. Fig.~\ref{fig:PS2F_a} illustrates the influence of the parameter $a$ with $M=5$ fixed. In line with Propositions~\ref{a0lemma} and~\ref{alemma}, increasing $a$ yields a  larger S\textsuperscript{2}-set. Fig.~\ref{fig:PS2F_M} shows the effect of the prediction horizon $M$ with $a=0.95$ fixed.
As predicted by Propositions~\ref{M1lemma} and~\ref{mlemma}, the S\textsuperscript{2}-set expands monotonically as $M$ increases. A longer predictive horizon provides more degrees of freedom for constructing safe–stable trajectories, thereby reducing conservatism and enlarging the set of admissible filtered inputs. 
A concise summary of these parameter effects is provided in Table~\ref{tab:PS2F_parameters_simple}.

\begin{table}[ht]
\centering
\caption{Effect of increasing each PS\textsuperscript{2}F design parameter.}
\vspace{1mm}

\renewcommand{\arraystretch}{1.2}

\begin{tabular}{m{1cm}|m{2.4cm}|m{2.8cm}}
\hline
\centering \textbf{Param.} 
& \centering \textbf{Effect on $\mathbb{S}(x)$} 
& \centering \textbf{Side Effect} \tabularnewline
\hline

\centering $a$ 
& \centering Enlarges $\mathbb{S}(x)$  
& \centering Typically slower convergence \tabularnewline
\hline

\centering $M$ 
& \centering Enlarges $\mathbb{S}(x)$  
& \centering Higher computational cost \tabularnewline
\hline

\end{tabular}
\label{tab:PS2F_parameters_simple}
\end{table}

We next investigate how the design parameters of the nominal MPC, which forms the 
foundation of PS\textsuperscript{2}F, affect the geometry of the resulting 
S\textsuperscript{2}-set. 
Figure~\ref{fig:PS2F_QR} illustrates the influence of the MPC state-weighting matrix, 
where we vary $Q = \rho I_2$ while fixing $R = I_2$ and $M = N = 5$. 
These results show how shaping the nominal cost landscape directly alters the boundary 
of the safe--stable domain. 
Figure~\ref{fig:PS2F_N} then examines the effect of the MPC prediction horizon $N$, 
with $M = N$ fixed. 
As $N$ increases, the S\textsuperscript{2}-set expands because the region of attraction 
of the nominal MPC becomes larger, enabling the filter to preserve a wider range of 
external commands while still guaranteeing safety and stability. 
Notably, when $N = 1$, the nominal MPC problem becomes infeasible, highlighting the 
critical role of the prediction horizon in ensuring feasibility of the overall 
PS\textsuperscript{2}F framework.

\begin{figure}
    \centering
    \includegraphics[width=0.5\linewidth]{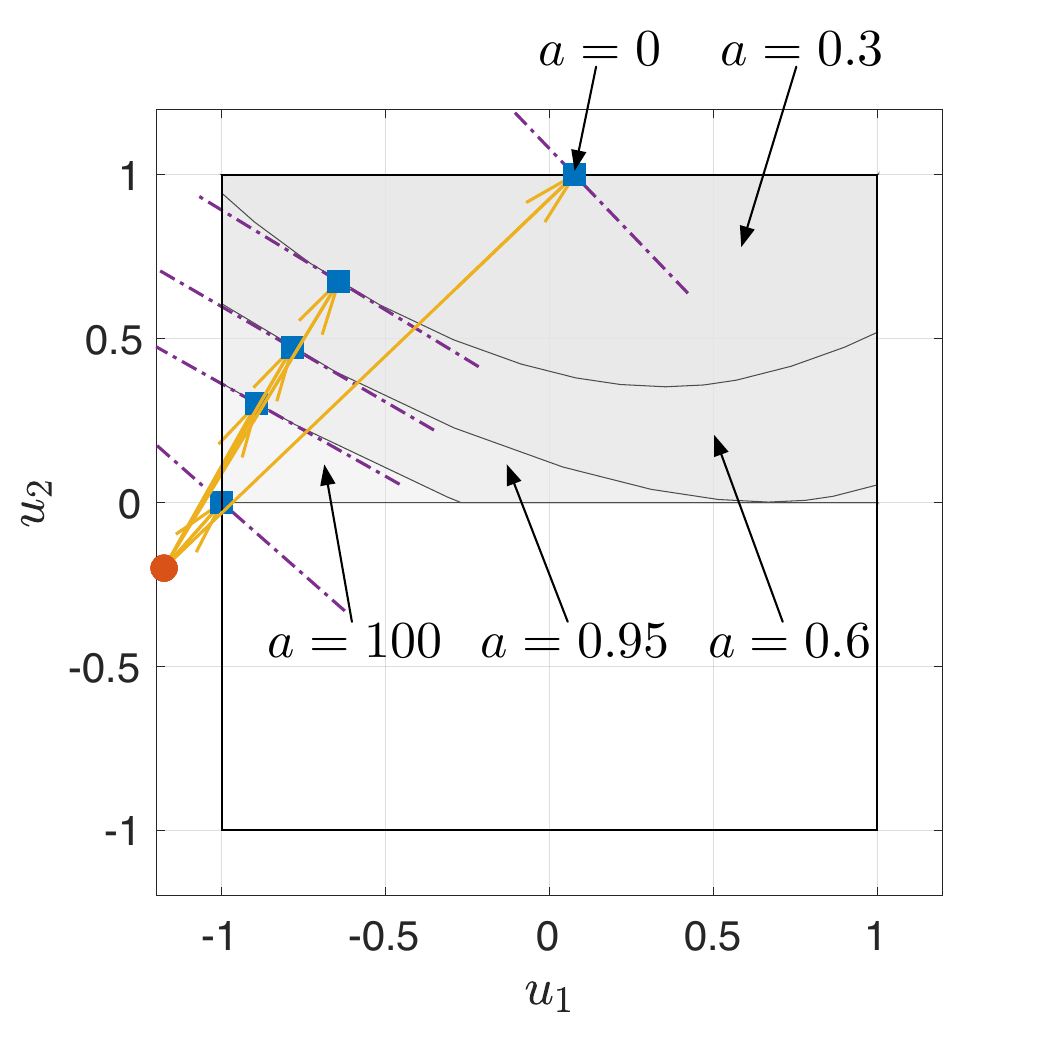}
    \caption{Effect of parameter $a$ on $\mathbb{S}(x)$, with $N=M=5$, $Q = 10I_2$, and  $R = I_2$.}
    \label{fig:PS2F_a}
\end{figure}

\begin{figure}
    \centering
    \includegraphics[width=0.5\linewidth]{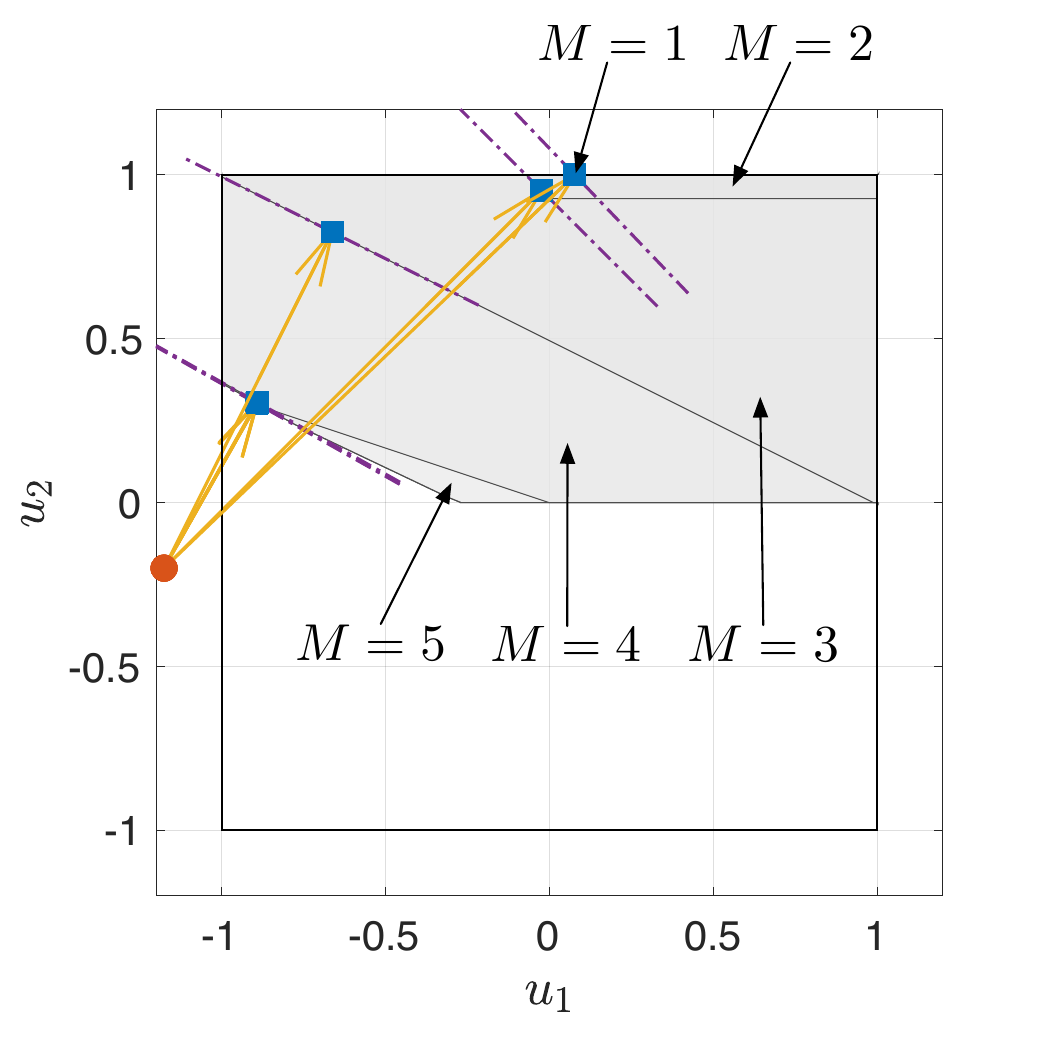}
    \caption{Effect of  prediction horizon  $M$ on $\mathbb{S}(x)$, with $N=5$, $a=0.95$, $Q = 10I_2$, and $R = I_2$.}
    \label{fig:PS2F_M}
\end{figure}

\begin{figure}
    \centering
    \includegraphics[width=0.5\linewidth]{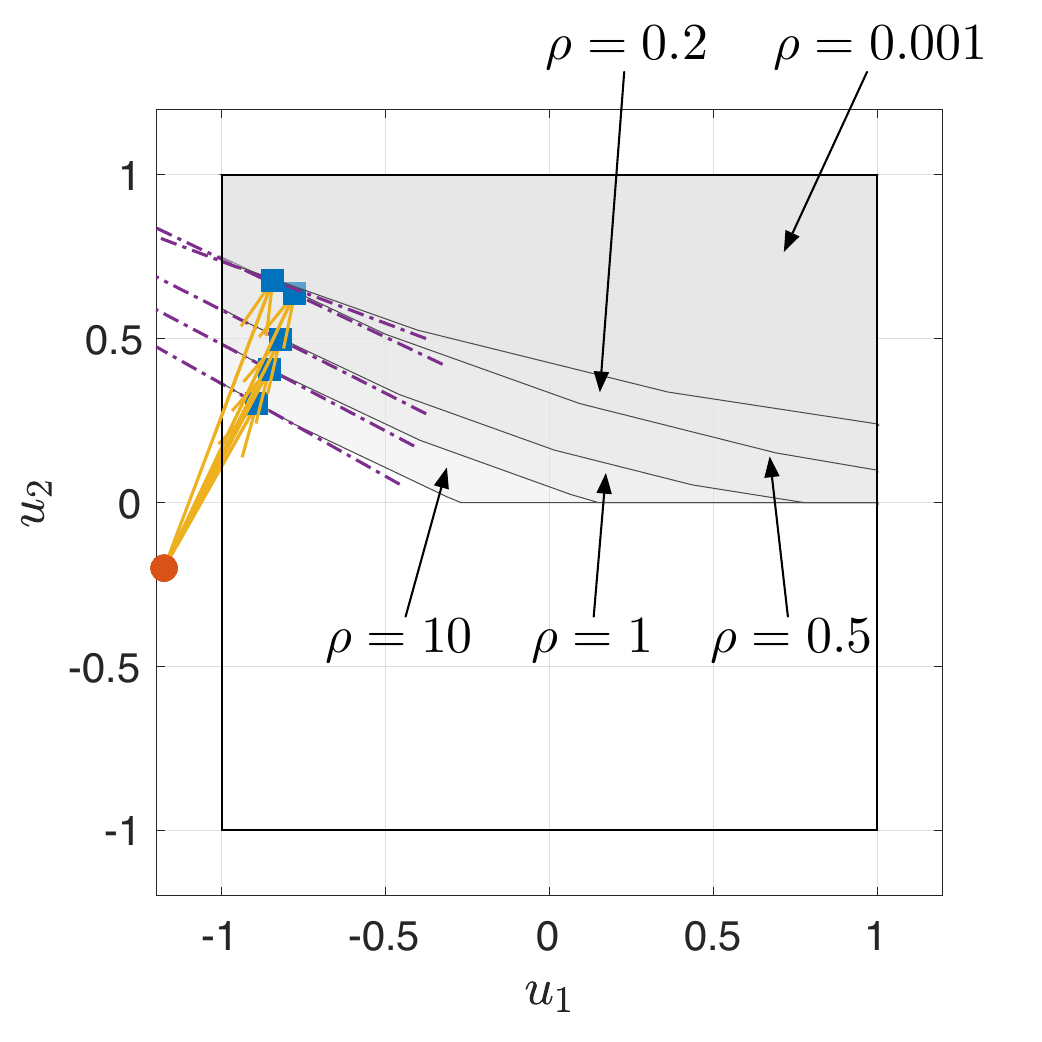}
    \caption{Effect of MPC weighting $Q=\rho I_2$ on $\mathbb{S}(x)$, with  $N=M=5$, $a=0.95$, and $R=I_2$.}
    \label{fig:PS2F_QR}
\end{figure}

\begin{figure}
    \centering
    \includegraphics[width=0.5\linewidth]{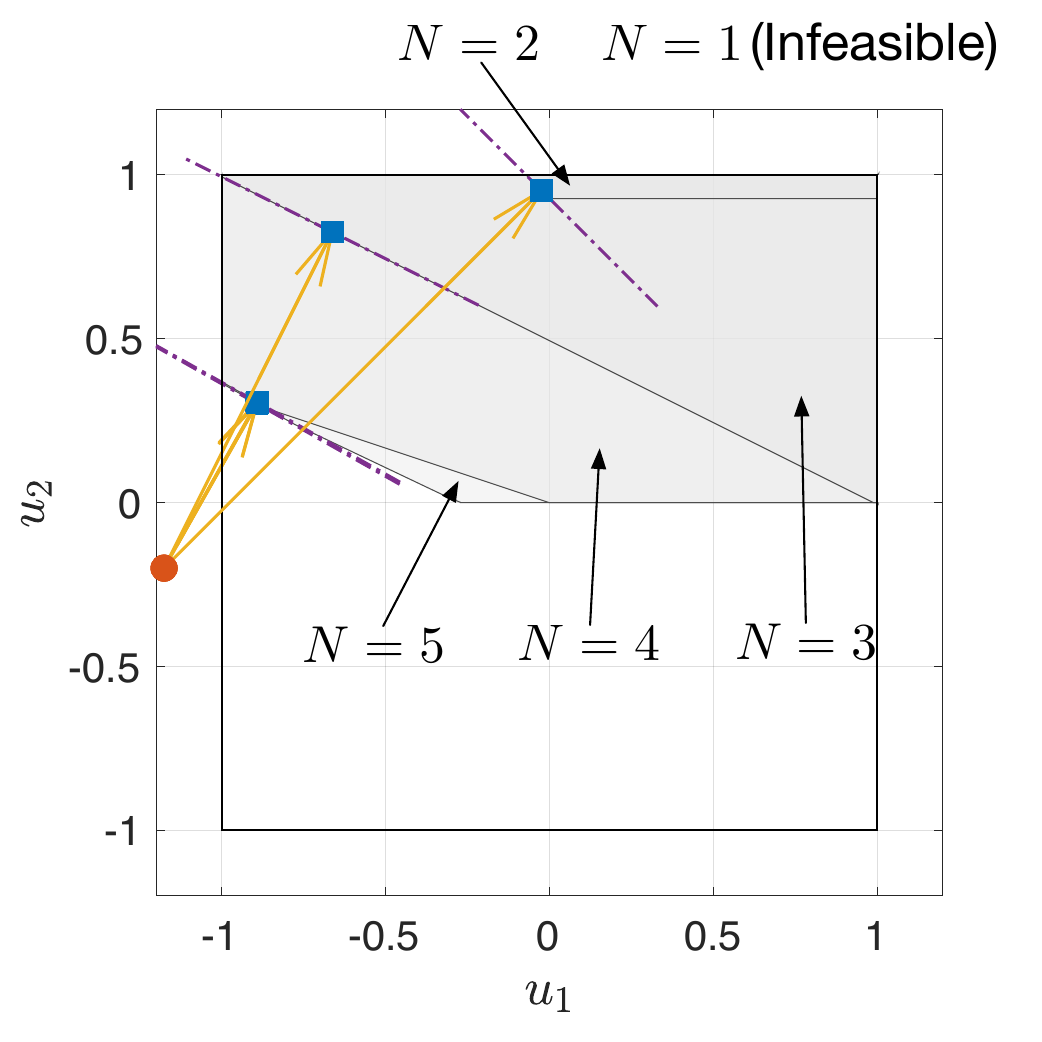}
    \caption{Effect of prediction horizon $N$ on $\mathbb{S}(x)$, with $Q = 10I_2$, $R = I_2$, $a=0.95$, and $M=N$.}
    \label{fig:PS2F_N}
\end{figure}

\subsection{Case Study 3: Go–Stay–Return Navigation}

In this final case study, we demonstrate the applicability of the 
PS\textsuperscript{2}F framework in a dynamic navigation task that 
requires both transient manoeuvring and temporary station-keeping. 
We consider a nonholonomic mobile robot with unicycle-type kinematics 
\cite{xie2021disturbance,kim2025nonholonomic}, whose dynamics are given by
\begin{equation}\label{eq:unicycle}
\begin{bmatrix}
    p_x(k+1)\\
    p_y(k+1)\\
    \theta(k+1)
\end{bmatrix}=\begin{bmatrix}
    p_x(k)\\
    p_y(k)\\
    \theta(k)
\end{bmatrix}+T_s \begin{bmatrix}
    \cos(\theta(k)) & 0\\
    \sin(\theta(k)) & 0\\
    0&1
\end{bmatrix} \begin{bmatrix}
    v(k)\\
    \omega(k)
\end{bmatrix}
\end{equation}
where $(p_x(k), p_y(k))$ denotes the robot’s planar position, $\theta(k)$ its heading angle, $v(k)$ the forward linear velocity, $\omega(k)$ the angular velocity, and $T_s = 0.2\,\mathrm{s}$  is the sampling time.
Let $x(k) = [p_x(k),\, p_y(k),\, \theta(k)]^\top$ and 
$u(k) = [v(k),\, \omega(k)]^\top$ denote the state and input vectors, 
respectively, with the initial condition $x(0) = 0_{3\times 1}$.
This model captures the fundamental nonholonomic constraint of wheeled 
robots: the robot moves only in the direction of its heading and must 
reorient via $\omega(k)$ to change direction. 
The unicycle model is widely used for ground robots, differential-drive 
vehicles, and autonomous mobile platforms operating in planar environments.

The task in this case study is to guide the mobile robot~\eqref{eq:unicycle}
from its initial position to a designated target region, remain within that region
for a prescribed dwell time, and subsequently return safely to its original location.
This three-phase “go–stay–return’’ behaviour is representative of many practical
robotic missions, including inspection–pause–return tasks, waypoint surveillance,
and object pickup followed by retreat \cite{nascimento2018nonholonomic}. Such tasks require not only accurate motion
execution during the outbound and return phases, but also the ability to maintain a
stable stationary posture within the target region. This combination of manoeuvring,
dwelling, and safe return makes the scenario particularly well-suited for evaluating
the mode scheduling capability of the PS\textsuperscript{2}F framework.
It is worth noting that this “go–stay–return’’ task cannot be reliably accomplished
by any single conventional controller, including a standard MPC design, because the
mission inherently involves distinct and sometimes conflicting control objectives
across phases, ranging from aggressive motion to long-duration stabilisation. To make the task even more challenging, we impose state and input constraints:
\[
\mathbb{X} = [-0.5,\,0.5]^2 \times [-\pi/3,\,\pi/3], 
\qquad 
\mathbb{U} = [-10,\,10]^2.
\]
 Without loss of generality, the collision volume of the robot is ignored here. The designated target that the robot must reach and remain is the point $p_{\text{go}}=[0.5,\,0.5]^\top$, the upper-right corner of the allowable position set, making the dwelling phase particularly demanding due to the tight state constraints. 

\subsubsection{The Baseline Approach}

A common baseline solution for this type of task is to design two separate controllers, each dedicated to one phase of the operation: a controller that drives the robot toward the goal (the ``go--stay'' phase) and another that brings it back to the origin (the ``return'' phase). These controllers operate independently and are switched sequentially to complete the overall mission.  

To generate these task-oriented commands, we adopt a reinforcement-learning-like optimal control strategy that prioritises goal achievement. Let
\begin{equation}\label{disgo}
D_{\mathrm{go}}(i;x):= \bigl| [p_x(i;x),\, p_y(i;x)]^\top - p_{\mathrm{go}} \bigr|^2
\end{equation}
denote the squared distance between the predicted robot position at step $i$ and the target point $p_{\mathrm{go}}$.  
The controller for the ``go-stay'' phase is given by
\begin{equation}\label{ugo}
\begin{aligned}
\mathbf{u}^*(x(k))
&= \{u^*(0;x(k)),\ldots,u^*(H-1;x(k))\}  = \arg\min_{\mathbf{u}} \sum_{i=0}^{H-1} \gamma^{\,i}\, D_{\mathrm{go}}(i;x(k)), \\
u_{\mathrm{go}}(k) &= u^*(0;x(k)),
\end{aligned}
\end{equation}
where $H=5$ is the prediction horizon and $\gamma=0.9\in(0,1)$ is a discount factor. Similarly, the ``return'' phase controller is obtained by replacing the goal position $p_{\mathrm{go}}$ with the home position 
$p_{\mathrm{return}} = [0,\,0]^\top$ in \eqref{disgo} and \eqref{ugo}. The resulting control law drives the robot back toward the origin using the same discounted predictive criterion, and its first control action is taken as the command for the return phase, denoted by $u_{\mathrm{return}}(k)$. 

Thus, under the conventional two-controller framework, the robot first applies $u_{\mathrm{go}}(k)$ during the ``go--stay" phase and subsequently switches to $u_{\mathrm{return}}(k)$ for the ``return" phase.

\subsubsection{PS\textsuperscript{2}F with Mode Scheduling}

This dynamic task requirement can also be naturally integrated into the PS\textsuperscript{2}F framework through Algorithm~\ref{alg:ps2f_complex}. 

In this demonstration, we only employ the external command $u_{\mathrm{ext}}(k) = u_{\mathrm{go}}(k)$ and omit the return controller $u_{\mathrm{return}}(k)$.  The nominal MPC employs the quadratic stage cost $\ell(x,u) = x^\top Q x + u^\top R u$
with tuning matrices \(Q = 10I_3\) and \(R = I_2\), prediction horizon \(N = 5\), terminal cost \(V_f(x)=0\), and terminal set \(\mathbb{X}_f = \{0_{3\times 1}\}\).  For the PS\textsuperscript{2}F layer, the scheduling parameters are chosen as \(M(k) = N=5\) for all \(k \in \mathbb{I}_{\ge 0}\), together with a time-varying  coefficient, given by
\[
a(k)=
\begin{cases}
100, &  k \in \mathbb{I}_{0: K_s-1},\\[2mm]
0.5, & k \in \mathbb{I}_{\ge K_s},
\end{cases}
\]
so that the closed-loop system prioritises safety during the initial phase (\(k \in \mathbb{I}_{0: K_s-1}\)) and transitions to a stability-oriented behaviour once the system is sufficiently close to the desired operating region. This scheduled adjustment enables the PS\textsuperscript{2}F mechanism to balance short-term safety requirements with long-term stability guarantees.

\subsubsection{Results}

The state and input responses are shown in Fig.~\ref{fig:robot_results}. 
For the baseline approach, the controller is scheduled to switch to the return mode at $k = 30$ (i.e., at $t = 6\,\mathrm{s}$), which corresponds to the blue curve in the filtered case in Fig.~\ref{fig:robot_results}. 
It is worth noting that the filtered case is slower than the baseline because it must keep the heading angle within the admissible safety region. This constraint requires additional time for the mobile robot to reorient before proceeding, as illustrated in Fig.~\ref{fig:robot_results}(a). 
From Fig.~\ref{fig:robot_results}, it can also be observed that the proposed method guarantees that the mobile robot safely completes the entire ``go--stay--return'' task, regardless of the choice of $K_s$.  Fig.~\ref{fig:robot_results}(c) shows the value function of the nominal MPC. As expected, it is nondecreasing before $K_s$ because, during the ``go--stay'' phase, the robot is actively moving away from the origin, causing the value function to increase (or at least not decrease). After $K_s$, according to Theorem~\ref{thm:timevarying}, the value function becomes decreasing, reflecting the transition to the stability-oriented phase of the controller.

To provide a more intuitive illustration, Fig.~\ref{fig:robot_xy} displays the robot’s 
position trajectories, where the triangular markers indicate the instantaneous heading 
angle. In Fig.~\ref{fig:robot_xy}(a), the baseline controller violates the 
heading-angle safety constraint when attempting to initiate motion near $k = 0$ or to 
return near $k = 30$, resulting in unsafe behaviour. In contrast, 
Fig.~\ref{fig:robot_xy}(b) shows that the proposed PS\textsuperscript{2}F method 
successfully maintains the robot within the admissible region while following the same 
switching schedule, demonstrating the effectiveness of the safety filter in enforcing 
constraint satisfaction throughout the manoeuvre.

\begin{figure}[ht]
    \centering
    \includegraphics[width=0.5\linewidth]{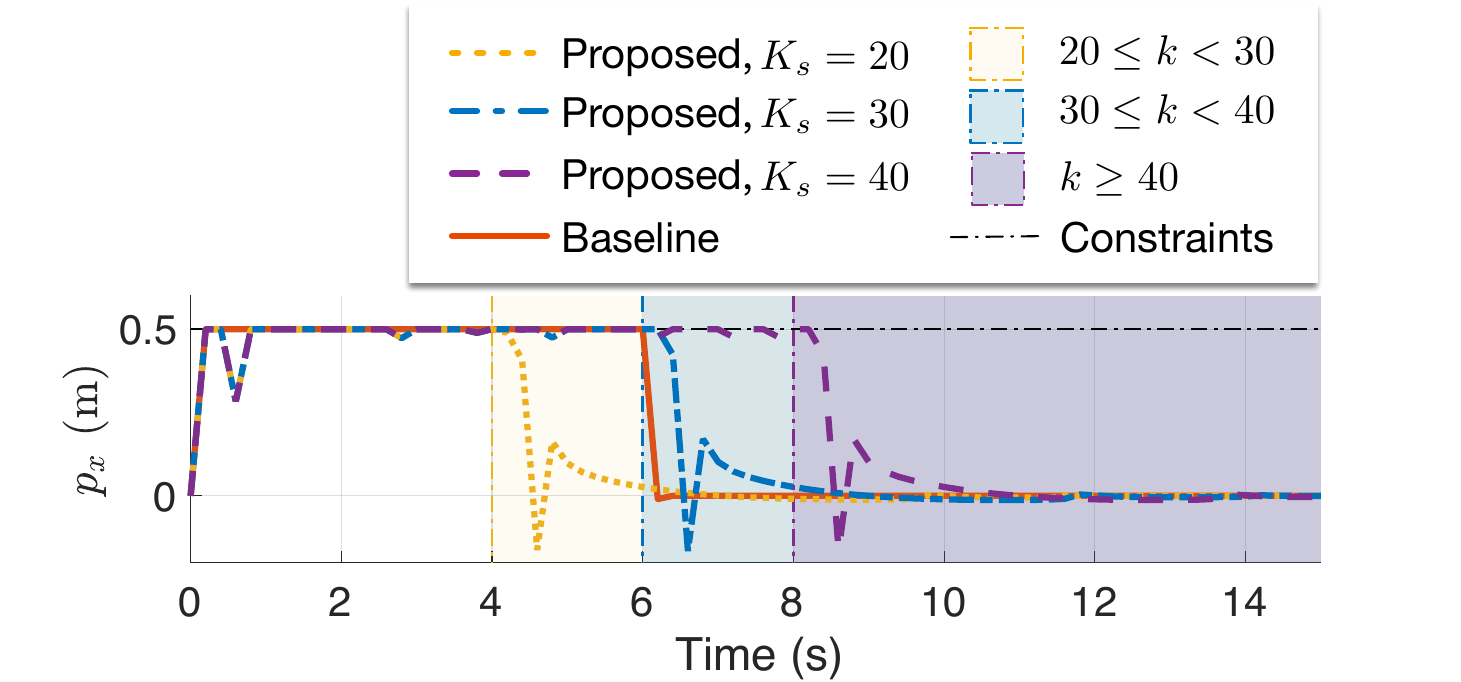}\\
     \includegraphics[width=0.5\linewidth]{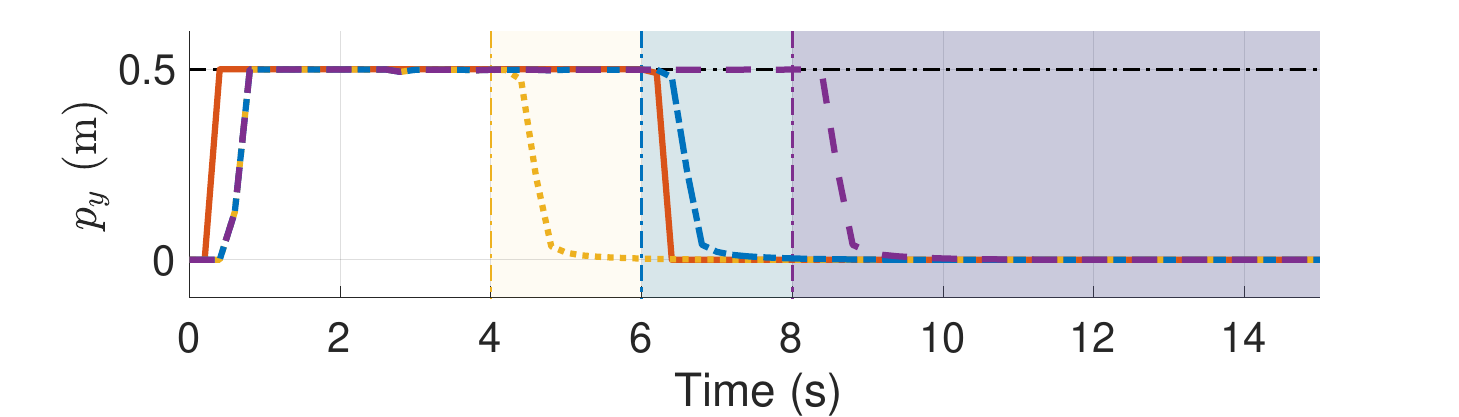}\\
     \includegraphics[width=0.5\linewidth]{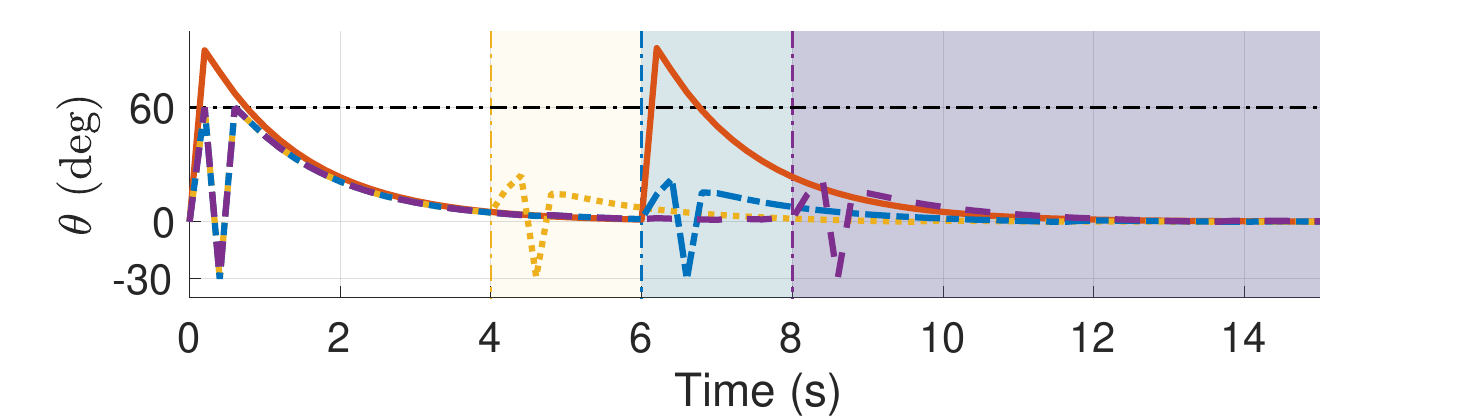}\\
     (a)\\
       \includegraphics[width=0.5\linewidth]{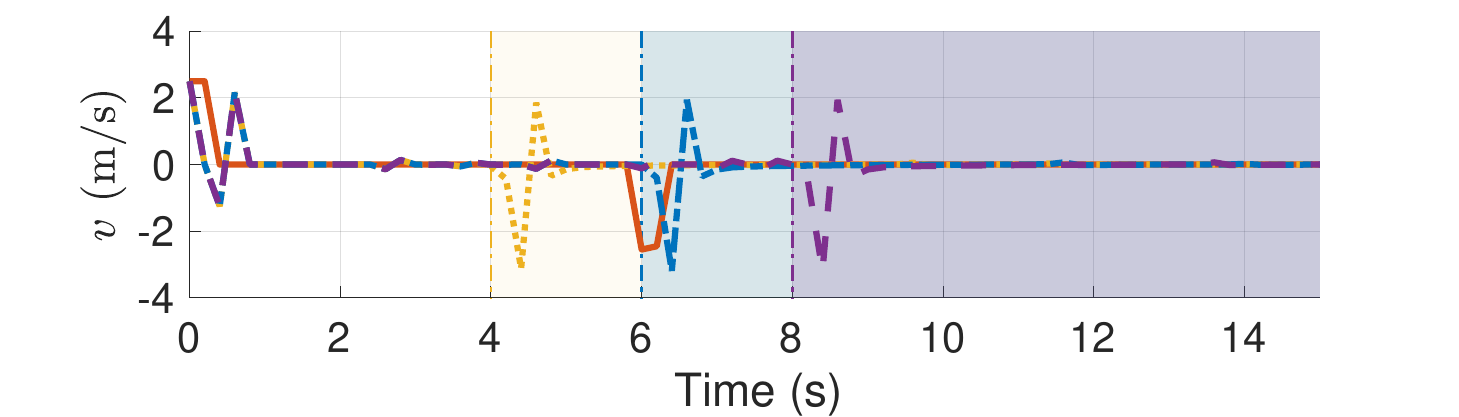}\\
     \includegraphics[width=0.5\linewidth]{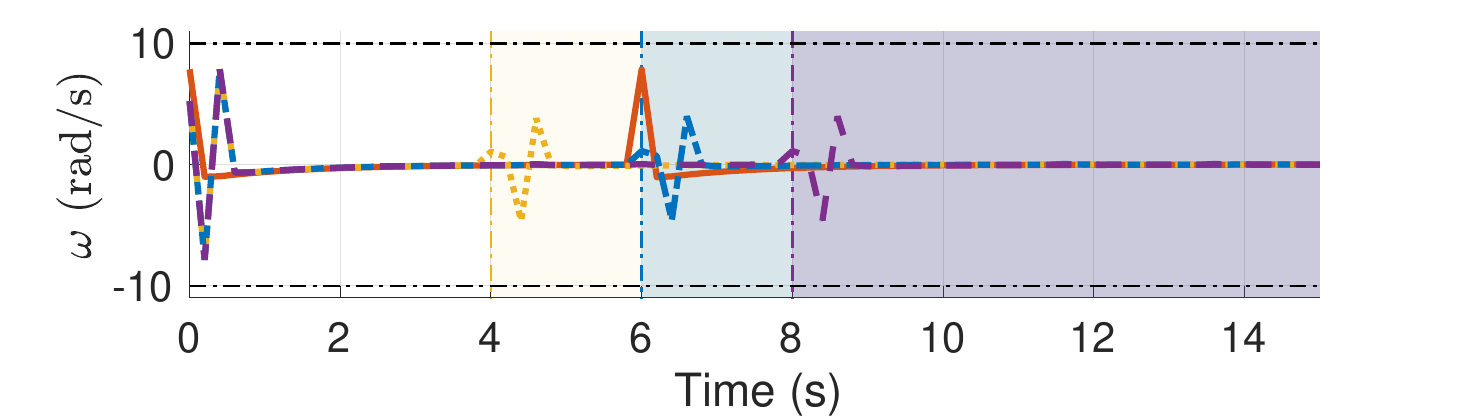}\\
     (b)\\
      \includegraphics[width=0.5\linewidth]{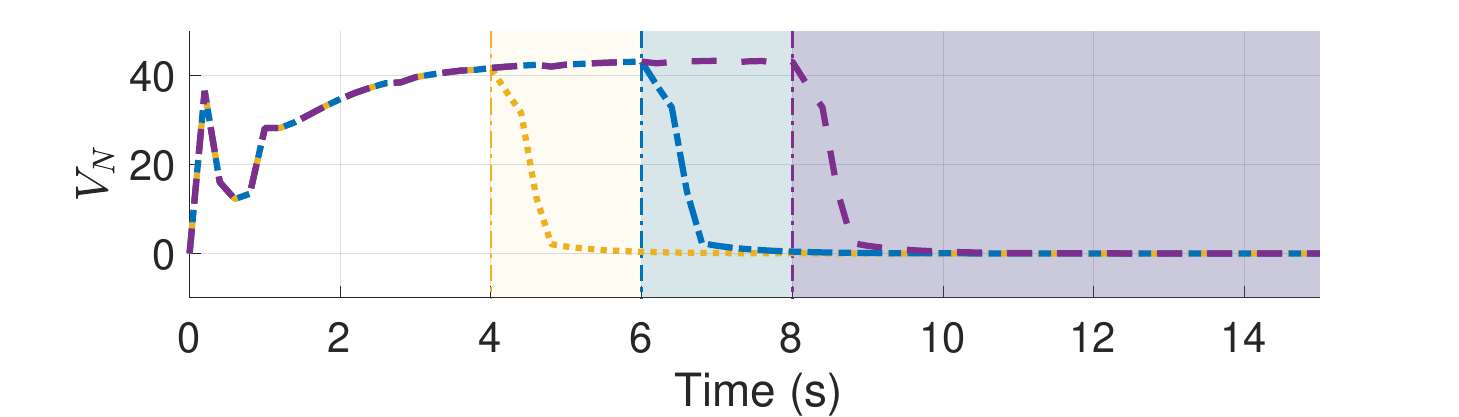}\\
     (c)
    \caption{Closed-loop evolution of the robot system under the baseline controller 
    and the proposed PS\textsuperscript{2}F method. 
    (a) States; (b) Control inputs; (c) Value function.}
    \label{fig:robot_results}
\end{figure}

\begin{figure}[ht]
    \centering
    \setlength{\tabcolsep}{-4pt} 
    \begin{tabular}{c}
\includegraphics[width=0.4\linewidth]{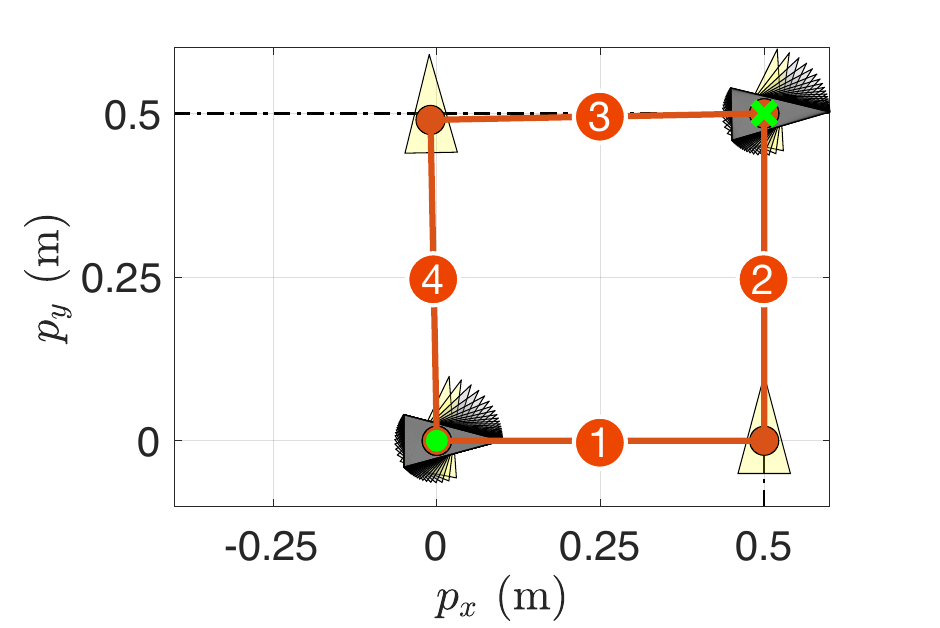}\\[-4pt]
  (a) \\ 
\includegraphics[width=0.4\linewidth]{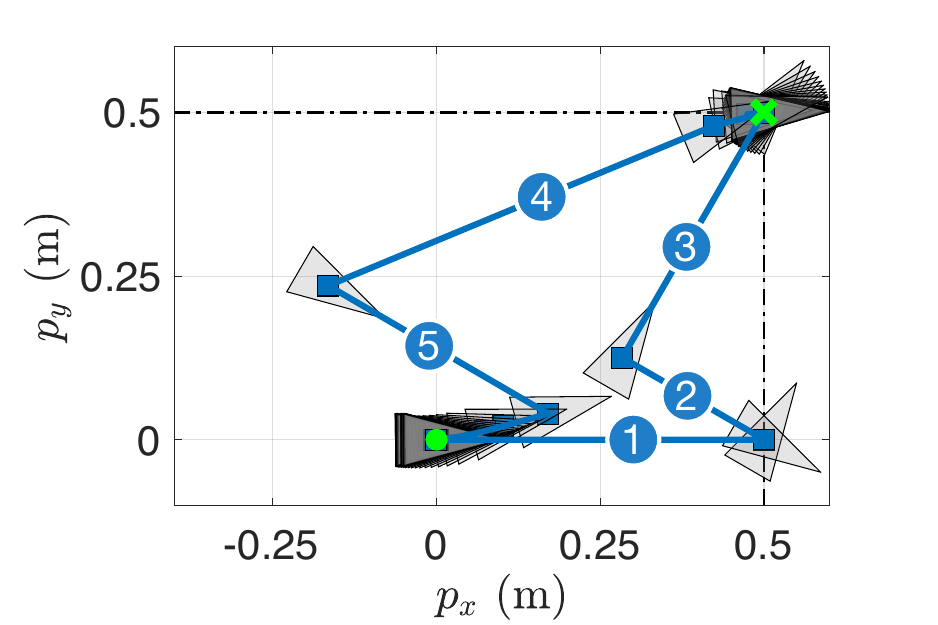} \\[-4pt]
  (b)  
    \end{tabular}

    \caption{Comparison of robot trajectories under the baseline controller
    and the proposed PS\textsuperscript{2}F method when both are scheduled to return at $k = 30$. A yellow robot marker denotes that the state violates
the safety constraints.
    (a) Baseline. (b) Proposed.} 
    \label{fig:robot_xy}
\end{figure}

\section{Conclusions}

In this paper, we have developed a unified predictive control framework for achieving safe and stable control under potentially unsafe or unstable external commands. 
Unlike most existing filter or safety-layer designs, the proposed \textit{Predictive Safety--Stability Filter} (PS\textsuperscript{2}F) provides a principled yet flexible mechanism that guarantees both safety and stability through a cascaded optimisation structure. 
By employing a nominal MPC layer as a \textit{copilot}, PS\textsuperscript{2}F inherits the theoretical guarantees of recursive feasibility and asymptotic stability while still allowing the incorporation of arbitrary goal-oriented external commands. Moreover, a time-varying parameterisation enables smooth mode scheduling, allowing the controller to prioritise safety during exploration and progressively enforce stability during exploitation. 
Numerical studies have demonstrated the effectiveness and versatility of the proposed framework across diverse operating scenarios. Future work will focus on extending the framework to uncertain, disturbed, or stochastic environments, and on implementing PS\textsuperscript{2}F on real robotic platforms to further validate its practical applicability.

\section*{Appendix} 
 \appendix

The appendix collects the detailed proofs: Appendix~\ref{proof_propos} contains the proofs of the propositions related to the second OCP, $\mathbb{P}_{f,M}(x)$, while Appendix~\ref{proof_theors} presents the proofs concerning the closed-loop behaviour of the overall system.

 \section{Proof of Propositions}\label{proof_propos}
\textbf{Proof of Proposition \ref{lem_fea}.}
Since $x \in \mathcal{X}_N$, the nominal MPC problem $\mathbb{P}_N(x)$ is feasible, and hence there exist optimal sequences $\mathbf{v}^*(x)$ and $\mathbf{z}^*(x)$. 
Consider the truncated sequences formed by taking the first $M{+}1$ elements of $\mathbf{z}^*(x)$ and the first $M$ elements of $\mathbf{v}^*(x)$, i.e.,
\begin{equation}
    \begin{aligned}
        \tilde{\mathbf{x}}(x) &:= \{ z^*(0;x), \dots, z^*(M;x) \}, \\
        \tilde{\mathbf{u}}(x) &:= \{ v^*(0;x), \dots, v^*(M-1;x) \}.
    \end{aligned}
\end{equation}
It is straightforward to verify that $\tilde{\mathbf{x}}(x)$ and $\tilde{\mathbf{u}}(x)$ satisfy \eqref{percons} and \eqref{mequ}; therefore, $\mathbb{P}_{f,M}(x)$ is feasible.  This completes the proof.
\QEDA

\textbf{Proof of Proposition \ref{a0lemma}.}
The first property is proved by contradiction.  
Assume that 
$$L(\mathbf{x}^*_{0:M-1} (x), \mathbf{u}^*(x)) < L(\mathbf{z}^*_{0:M-1}(x), \mathbf{v}^*_{0:M-1}(x)).$$  
Consider the extended sequences constructed by appending the tail of $\mathbf{z}^*(x)$ and $\mathbf{v}^*(x)$ to $\mathbf{x}^*(x)$ and $\mathbf{u}^*(x)$, respectively:
\begin{equation*} 
    \begin{aligned}
     \tilde{\mathbf{z}}(x) &:= \{x^*(0; x), \dots, x^*(M; x),  z^*(M+1; x), \dots, z^*(N; x)\}\\
     \tilde{\mathbf{v}}(x) & := \{u^*(0; x), \dots,  u^*(M-1; x),  v^*(M; x), \dots, v^*(N-1; x)\} 
     \end{aligned}
\end{equation*}
Since $z^*(M;x) = x^*(M;x)$ from~\eqref{mequ}, 
the sequence $\tilde{\mathbf{z}}(x)$ is the state trajectory induced by the input sequence $\tilde{\mathbf{v}}(x)$.  
Both are feasible for $\mathbb{P}_N(x)$, and moreover,
\begin{equation*}
\begin{aligned}
V_N(x, \tilde{\mathbf{v}}(x))
   &= \sum_{i=0}^{M-1} \ell(x^*(i; x), u^*(i; x)) + \sum_{i=M}^{N-1} \ell(z^*(i; x), v^*(i; x))
    + V_f(z^*(N; x)) \\
   &= L(\mathbf{x}^*_{0:M-1}(x), \mathbf{u}^* (x))  +    V_N(x, \mathbf{v}^*(x)) - L(\mathbf{z}^*_{0:M-1}(x), \mathbf{v}^*_{0:M-1}(x))   \\
   &< V_N(x, \mathbf{v}^*(x)) =V_N^*(x)
\end{aligned}
\end{equation*}
which contradicts the optimality of $\mathbf{v}^*(x)$ for $\mathbb{P}_N(x)$.  
Hence, $L(\mathbf{x}^*_{0:M-1}(x), \mathbf{u}^*(x)) \ge L(\mathbf{z}^*_{0:M-1}(x), \mathbf{v}^*_{0:M-1}(x))$.  
By noting~\eqref{percons}, equality follows: 
$$L(\mathbf{x}^*_{0:M-1}(x), \mathbf{u}^*(x)) = L(\mathbf{z}^*_{0:M-1}(x), \mathbf{v}^*_{0:M-1}(x)).$$

With property~(a) established, we have 
$V_N(x, \tilde{\mathbf{v}}(x))  = V_N(x, \mathbf{v}^*(x))$.  
If $\mathbb{P}_N(x)$ is strictly convex, this implies $\tilde{\mathbf{v}} (x)= \mathbf{v}^*(x)$, 
and thus $u = u^*(0;x) = v^*(0;x)$. This completes the proof. 
\QEDA

\textbf{Proof of Proposition \ref{alemma}.}
Consider $a = a_2$.  
For any element $u(0;x) \in \mathbb{S}^{a_2}(x)$, there exists a control sequence
\[
\mathbf{u}(x) = \{ u(0;x), u(1;x), \dots, u(M{-}1;x) \} \in \mathcal{U}^{a_2}_M(x)
\]
such that the corresponding state trajectory satisfies the inequality
\begin{equation*} 
\begin{aligned}
    &L(\mathbf{x}_{0:M-1}(x), \mathbf{u}(x)) 
    - L(\mathbf{z}^*_{0:M-1}(x), \mathbf{v}^*_{0:M-1}(x))   \le a_2\,\ell(x(0;x), u(0;x))\le a_1\,\ell(x(0;x), u(0;x))
\end{aligned}
\end{equation*}
Hence,  $\mathbf{u}(x) \in \mathcal{U}^{a_1}_M(x)$, which implies that 
$u(0; x) \in \mathbb{S}^{a_1}(x)$. 
Therefore, $\mathbb{S}^{a_2}(x) \subseteq \mathbb{S}^{a_1}(x)$. This completes the proof.  
\QEDA

\textbf{Proof of Proposition \ref{M1lemma}.}
For $M=1$, the terminal equality constraint \eqref{mequ} reduces to $x(1;x) = z^*(1;x)$. 
By the system dynamics, this is equivalent to
$f\bigl(x, u^*(0;x)\bigr) = f\bigl(x, v^*(0;x)\bigr)$.
Since the map $u \mapsto f(x,u)$ is injective for the given state $x$, the above equality
implies $u^*(0;x) = v^*(0;x)$. This completes the proof.
\QEDA

\textbf{Proof of Proposition \ref{mlemma}.}  Consider $M = i$. For any element $u(0; x) \in \mathbb{S}_i(x)$, there exists a control sequence 
\[
\mathbf{u}_i(x) = \{ u(0; x), u(1; x), \dots, u(i-1; x) \} \in \mathcal{U}^a_i(x)
\]
such that the following hold:
\begin{subequations}
\begin{align}
    &L(\mathbf{x}_{0:i-1}(x), \mathbf{u}_i(x)) - L(\mathbf{z}^*_{0:i-1}(x), \mathbf{v}^*_{0:i-1}(x)) - a\ell(x(0;x), u(0;x)) \le 0, \label{percons2}\\ 
    &x(i;x) = z^*(i; x). \label{mequ2}
\end{align}
\end{subequations}
Using~\eqref{mequ2}, we can extend the control sequence $\mathbf{u}_i(x)$ by appending the optimal input $v^*(i; x)$ to obtain
\[
\tilde{\mathbf{u}}_{i+1}(x): = \{ u(0; x), u(1; x), \dots, u(i-1; x), v^*(i; x) \}.
\]
The corresponding state sequence is
\[
\tilde{\mathbf{x}}(x): = \{ x(0; x),  \dots, x(i-1; x), z^*(i; x), z^*(i+1; x) \}.
\]
The constructed pair $(\tilde{\mathbf{x}}(x), \tilde{\mathbf{u}}_{i+1}(x))$ satisfies both~\eqref{percons2} and~\eqref{mequ2} with $i$ replaced by $i+1$. 
Hence, $\tilde{\mathbf{u}}_{i+1}(x) \in \mathcal{U}^a_{i+1}(x)$, which implies that 
$u(0; x) \in \mathbb{S}_{i+1}(x)$. 
Therefore, $\mathbb{S}_i(x) \subseteq \mathbb{S}_{i+1}(x)$ for all $i \in \mathbb{I}_{1:N-1}$. This completes the proof.  
\QEDA

\textbf{Proof of Proposition~\ref{linearsys}.}
Since the safety constraints in both OCPs, as well as the terminal constraint of the
first OCP, are inactive, the nominal MPC problem coincides with the infinite-horizon
LQR problem. Hence, the value function of the nominal MPC is $V_N^*(x) = x^\top P x$, and the optimal state and input trajectories are given by
\begin{equation}\label{eq:LQR_traj}
\begin{aligned}
    z^*(i;x) &= A_{\mathrm{LQR}}^{\,i} x, 
        \qquad i = 0,1,\ldots, N, \\ 
    v^*(i;x) &= -K_{\mathrm{LQR}}\, z^*(i;x),
\end{aligned}
\end{equation}  
where $K_{\mathrm{LQR}}:= (R + B^\top P B)^{-1} B^\top P A$ qnd $A_{\mathrm{LQR}} := A - BK_{\mathrm{LQR}}$. With \eqref{eq:LQR_traj}, the cumulative stage cost of the nominal optimal trajectory over the first $M$ steps is
\begin{equation}\label{eq:L_nominal_segment_closed_form}
\begin{aligned}
   & L\bigl(\mathbf{z}^*_{0:M-1}(x), \mathbf{v}^*_{0:M-1}(x)\bigr)\\
    &= \sum_{i=0}^{M-1} \ell\left(z^*(i;x),v^*(i;x)\right)+V_f(z^*(M;x)) -V_f(z^*(M;x))\\
    &=z^*(0;x)^\top P z^*(0;x)- z^*(M;x)^\top P z^*(M;x)\\
    &= x^\top \left(P - A_{\mathrm{LQR}}^{M\top} P A_{\mathrm{LQR}}^{M}\right) x.
\end{aligned}
\end{equation}
Next, we derive the lifted form of the dynamics and rewrite the conditions 
\eqref{percons} and \eqref{mequ} accordingly.
Over a prediction horizon of length $M$, define the stacked state and input vectors
\[
\begin{aligned}
   \mathbf{x}:=
    \begin{bmatrix}
        x(0;x)  \\ \vdots \\ x(M;x)
    \end{bmatrix}
    \in \mathbb{R}^{(M+1)n},~
     \mathbf{u}
    :=
    \begin{bmatrix}
        u(0;x)  \\ \vdots \\ u(M-1;x)
    \end{bmatrix}
    \in \mathbb{R}^{Mm}.
\end{aligned}
\]
For notational simplicity, we use the same symbols $\mathbf{x}$ and $\mathbf{u}$ as in the
sequence representation of the second OCP, $\mathbb{P}_{f,M}(x)$; in this context, however, they denote the stacked state and
input vectors.
The system dynamics $x(i+1;x) = A x(i;x) + B u(i;x)$ imply the lifted (affine) dynamics
\begin{equation}\label{eq:lifted_dynamics}
    \mathbf{x} = \Phi_M x + \Gamma_M \mathbf{u},
\end{equation}
with $\Phi_M \in \mathbb{R}^{(M+1)n \times n}$ and $\Gamma_M\in \mathbb{R}^{(M+1)n \times Mm}$, given by
\[
\begin{aligned}
    \Phi_M
    &:=
    \begin{bmatrix}
        I_n \\ 
        A \\ 
        A^2 \\
        \vdots \\
        A^{M}
    \end{bmatrix},~
    \Gamma_M
    :=
    \begin{bmatrix}
        0_{n\times m}      &  0_{n\times m}       & \cdots &  0_{n\times m} \\
        B      &  0_{n\times m}      & \cdots &  0_{n\times m} \\
        A B    & B      & \cdots &  0_{n\times m} \\
        \vdots & \vdots & \ddots & \vdots \\
        A^{M-1} B & A^{M-2} B & \cdots & B
    \end{bmatrix}.
\end{aligned}   
\]
In particular, the terminal state satisfies
\[
    x(M;x)= e_M \mathbf{x}
        = e_M \Phi_M x + e_M \Gamma_M \mathbf{u},
\]
where $e_M := [0_{n\times n},~\cdots,~0_{n\times n},~I_n] \in \mathbb{R}^{n\times (M+1)n}$ selects the last state component. The terminal  constraint \eqref{mequ} can be equivalently expressed as
\begin{equation}\label{eq:affine_constraint}
    A_{\mathrm{eq}} \mathbf{u} = b_{\mathrm{eq}} x
\end{equation}
where $ A_{\mathrm{eq}} := e_M \Gamma_M$ and $b_{\mathrm{eq}} := A_{\mathrm{LQR}}^M - e_M \Phi_M.$

To express the inequality \eqref{percons} in compact lifted form, we introduce the
block-diagonal matrices
\[
\begin{aligned}
 \tilde{Q} &:= \operatorname{diag}(Q,\dots,Q,0_{n\times n}) 
    \in \mathbb{R}^{(M+1)n\times (M+1)n},\\
 \tilde{R} &:= \operatorname{diag}(R,\dots,R) 
    \in \mathbb{R}^{Mm\times Mm},   
\end{aligned}
\]
so that the cumulative stage cost over the first $M$ steps can be written as
\[
    L(\mathbf{x}_{0:M-1}(x), \mathbf{u}(x))
    = \mathbf{x}^\top \tilde{Q}\,\mathbf{x}
      + \mathbf{u}^\top \tilde{R}\,\mathbf{u}.
\]
Using the lifted dynamics \eqref{eq:lifted_dynamics}, we obtain
\begin{equation}\label{eq:L_lifted}
\begin{aligned}
   & L(\mathbf{x}_{0:M-1}(x), \mathbf{u}(x))\\
    &= (\Phi_M x + \Gamma_M \mathbf{u})^\top 
       \tilde{Q}\,(\Phi_M x + \Gamma_M \mathbf{u})
       + \mathbf{u}^\top \tilde{R}\,\mathbf{u} \\[1mm]
    &= x^\top \Phi_M^\top \tilde{Q}\,\Phi_M x
       + 2 x^\top \Phi_M^\top \tilde{Q}\,\Gamma_M \mathbf{u}   + \mathbf{u}^\top 
        \bigl(\Gamma_M^\top \tilde{Q}\,\Gamma_M + \tilde{R}\bigr)\mathbf{u}.
\end{aligned}
\end{equation}
Similarly, the first-stage cost can be written as
\begin{equation}\label{eq:ell_stage0}
\begin{aligned}
   & \ell(x(0;x), u(0;x))\\
    &= x(0;x)^\top Q\, x(0;x)
       + u(0;x)^\top R\,u(0;x) \\
    &= x^\top Q\, x + \mathbf{u}^\top e_1^\top R e_1 \mathbf{u},
\end{aligned}
\end{equation}
where $e_1 :=[I_m, 0_{m\times m},\cdots, 0_{m\times m}] 
    \in \mathbb{R}^{m\times (Mm)}$
selects the first input block from $\mathbf{u}$. 
Substituting \eqref{eq:L_nominal_segment_closed_form}, \eqref{eq:L_lifted},  and \eqref{eq:ell_stage0} into \eqref{percons} yields
\begin{equation}\label{eq:affine_constraint2}
  \mathbf{u}^\top H\,\mathbf{u}
    + 2 x^\top F\,\mathbf{u}
    + x^\top G x \le 0,   
\end{equation}
where 
\[
\begin{aligned}
    H 
    &:= \Gamma_M^\top \tilde{Q}\,\Gamma_M 
        + \tilde{R} 
        - a\, e_1^\top R e_1, \\[1mm]
    F 
    &:= \Phi_M^\top \tilde{Q}\,\Gamma_M, \\[1mm]
    G 
    &:= \Phi_M^\top \tilde{Q}\,\Phi_M 
        - \bigl(P - A_{\mathrm{LQR}}^{M\top} P A_{\mathrm{LQR}}^{M}\bigr)
        - a Q.
\end{aligned}
\]
Together with \eqref{eq:affine_constraint} and \eqref{eq:affine_constraint2}, this
establishes the lifted quadratic representation of the PS\textsuperscript{2}F
constraints and completes the proof.
\QEDA

 \section{Proof of Theorems}\label{proof_theors}

\textbf{Proof of Theorem \ref{thm:Recursive feasibility and stability}.} To prove property~(a), we first establish recursive feasibility. 
At time step $k$, suppose that $x(k) \in \mathcal{X}_N$. 
Then, the nominal MPC problem $\mathbb{P}_N(x(k))$ is feasible, and the corresponding optimal sequences $\mathbf{v}^*(x(k))$ and $\mathbf{z}^*(x(k))$ exist, as defined in~\eqref{vstar} and~\eqref{zstar}. 
Since $z^*(N; x(k)) \in \mathbb{X}_f$, Assumption~\ref{assum:cost_bounds} ensures the existence of a control input $\tilde{v}(N; x(k))\in\mathbb{U}$ such that the successor state 
\begin{equation}
    \tilde{z}(N{+}1; x(k)) := f\big(z^*(N; x(k)), \tilde{v}(N; x(k))\big)
\end{equation}
satisfies $\tilde{z}(N{+}1; x(k)) \in \mathbb{X}_f$ and 
\begin{equation}\label{terminalcost}
\begin{aligned}
& V_f\big(\tilde{z}(N{+}1; x(k))\big) - V_f\big(z^*(N; x(k))\big)   \le -\,\ell\big(z^*(N; x(k)), \tilde{v}(N; x(k))\big).   
\end{aligned}
\end{equation}
By Proposition~\ref{lem_fea}, the filter problem $\mathbb{P}_{f,M}(x(k))$ is also feasible under the same condition, and its optimal sequences $\mathbf{u}^*(x(k))$ and $\mathbf{x}^*(x(k))$ exist as given in~\eqref{ustar} and~\eqref{xstar}.

At time step $k+1$, noting that $u(k) = u^*(0; x(k))$, the successor state satisfies $x(k+1) = f(x(k), u(k)) = x^*(1; x(k)).$ To establish recursive feasibility, consider the extended sequences constructed by shifting the optimal solution of $\mathbb{P}_{f,M}(x(k))$ one step forward and appending the terminal element of the nominal MPC trajectory. Specifically, define
\begin{equation}\label{feak1}
\begin{aligned}
&\tilde{\mathbf{v}}(x(k{+}1)) := 
\big\{ u^*(1; x(k)), \dots, u^*(M{-}1; x(k)),  v^*(M; x(k)), \dots, v^*(N{-}1; x(k)), \tilde{v}(N; x(k)) \big\},\\
&\tilde{\mathbf{z}}(x(k{+}1)) := 
\big\{ x^*(1; x(k)), \dots, x^*(M; x(k)),  z^*(M{+}1; x(k)), \dots, z^*(N; x(k)), \tilde{z}(N{+}1; x(k)) \big\}.
\end{aligned}
\end{equation}
Since $z^*(M; x(k)) = x^*(M; x(k))$ by~\eqref{mequ}, the sequence $\tilde{\mathbf{z}}(x(k+1))$ is precisely the state trajectory induced by the control sequence $\tilde{\mathbf{v}}(x(k+1))$. 
Therefore, the constructed pair $(\tilde{\mathbf{z}}(x(k+1)), \tilde{\mathbf{v}}(x(k{+}1)))$ is feasible for the nominal MPC problem $\mathbb{P}_N(x(k+1))$.

By induction, if $x(0) \in \mathcal{X}_N$, then $\mathbb{P}_N(x(0))$ is feasible at the initial step. 
From the construction above, feasibility of $\mathbb{P}_N(x(k))$ at time $k$ implies that both  $\mathbb{P}_{f,M}(x(k{+}1))$ and $\mathbb{P}_N(x(k{+}1))$  remain feasible. 
Hence, $\mathbb{P}_N(x(k))$ and $\mathbb{P}_{f,M}(x(k))$ are feasible for all $k \in\mathbb{I}_{\ge 0}$. The overall recursive feasibility process is illustrated in Fig.~\ref{fig:recursive feasibility}.
The feasibility of these optimisation problems directly ensures that the state and input constraints are satisfied at every time step. This completes the proof of property~(a). 

\begin{figure}[t]
\centering
    \centering
    \includegraphics[width=0.6\linewidth]{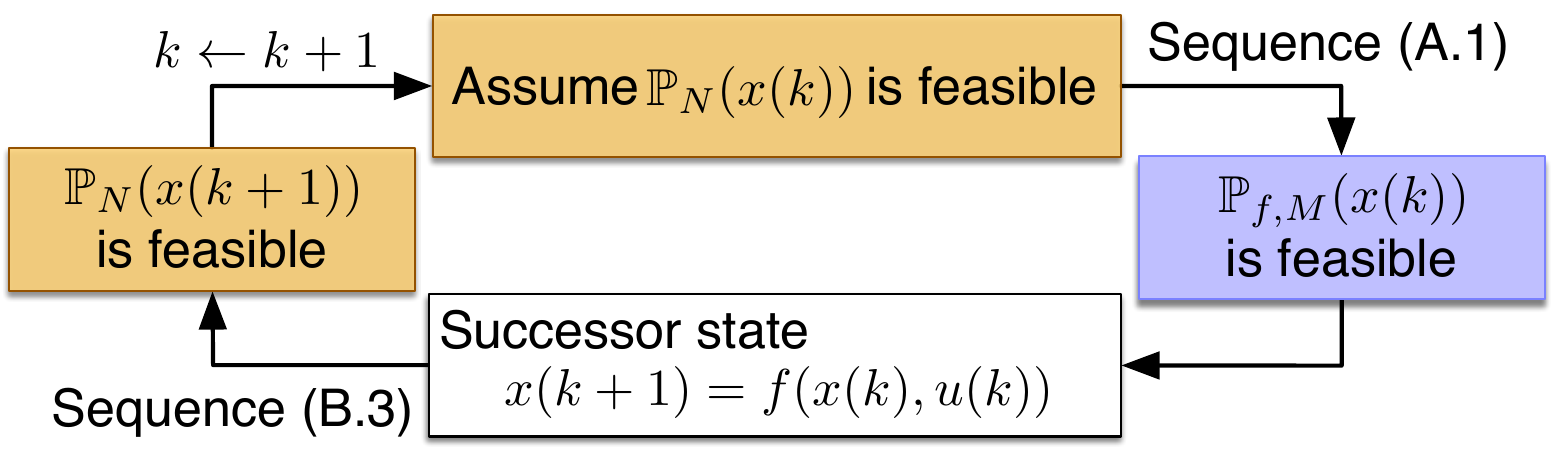}
    \caption{Schematic illustration of the recursive feasibility analysis.}
    \label{fig:recursive feasibility}
\end{figure}

To prove property~(b), we employ the optimal value function $V_N^*(x(k))$ of the nominal MPC problem as a Lyapunov function candidate.  At time step $k{+}1$, the value of the Lyapunov candidate must be no greater than the cost evaluated along the feasible sequences defined in~\eqref{feak1}. Hence,
\begin{equation}
\begin{aligned}
V^*_N(x(k+1)) &\le V_N(x(k+1),\tilde{\mathbf{v}}(x(k{+}1)) )\\
   &=\sum_{i=0}^{M-1} \ell(x^*(i;x(k))),u^*(i;x(k)))  -\ell(x^*(0;x(k)), u^*(0;x(k))) \\
   &~~~+ \sum_{i=M}^{N-1} \ell(z^*(i;x(k)),v^*(i;x(k)))  + \ell(z^*(N;x(k)),\tilde{v}(N; x(k)))  + V_f(\tilde{z}(N+1; x(k)))
\end{aligned}    
\end{equation}
Noting~\eqref{percons} and~\eqref{terminalcost}, it follows that
\begin{equation}\label{vdecrease}
 \begin{aligned}
 V^*_N(x(k+1))&\le L(\mathbf{z}^*_{0:M-1}(x(k)), \mathbf{v}^*_{0:M-1}(x(k))) -(1-a)\ell(x^*(0;x(k)), u^*(0;x(k)))\\
 &~~~+\sum_{i=M}^{N-1} \ell(z^*(i;x(k)),v^*(i;x(k))) +  V_f(z^*(N; x(k)))\\
 &=V^*_N(x(k)) -(1-a)\ell(x^*(0;x(k)), u^*(0;x(k)))\\
 &\le V^*_N(x(k)) -(1-a)\,\alpha_1(|x(k)|)
 \end{aligned}    
\end{equation}
where $\alpha_1(\cdot)$ is a \(\mathcal{K}_\infty\) function from Assumption~\ref{assum:cost_bounds}. By Assumption~\ref{assum:cost_bounds} and Proposition~\ref{lem:VN_upper}, the optimal value function satisfies
\begin{equation}\label{finalbounds}
    \alpha_1(|x(k)|) \le V_N^*(x(k)) \le \beta(|x(k)|).
\end{equation}
Together with the decrease condition in~\eqref{vdecrease},  the value function of the nominal MPC, $V_N^*(x(k))$, serves as a valid Lyapunov function for the closed-loop system, ensuring asymptotic stability of the origin. 
This completes the proof of property~(b). \QEDA

\textbf{Proof of Theorem \ref{thm:timevarying}.}
Property~(a) can be established by constructing feasible sequences in the same manner as in the proof of Theorem~\ref{thm:Recursive feasibility and stability}, with the fixed horizon $M$ replaced by the time-varying horizon $M(k)$. 

Since $\sup_{k \ge K_s} a(k) < 1$, there exists a constant $\bar{a} \in (0,1]$ such that 
$a(k) \le 1 - \bar{a}$ for all $k \in\mathbb{I}_{\ge K_s}$. 
It then follows that
\begin{equation}\label{vdecrease2}
\begin{aligned}
V_N^*(x(k{+}1)) 
&\le V_N^*(x(k)) - \big(1 - a(k)\big)\,\alpha_1\big(|x(k)|\big)  \le V_N^*(x(k)) - \bar{a}\,\alpha_1\big(|x(k)|\big), 
\quad \forall\, k \in\mathbb{I}_{\ge K_s},
\end{aligned}
\end{equation}
which ensures monotonic decrease of $V_N^*(x(k))$ and hence proves property~(b). 
\QEDA


\begin{thebibliography}{10}
\providecommand{\url}[1]{#1}
\csname url@samestyle\endcsname
\providecommand{\newblock}{\relax}
\providecommand{\bibinfo}[2]{#2}
\providecommand{\BIBentrySTDinterwordspacing}{\spaceskip=0pt\relax}
\providecommand{\BIBentryALTinterwordstretchfactor}{4}
\providecommand{\BIBentryALTinterwordspacing}{\spaceskip=\fontdimen2\font plus
\BIBentryALTinterwordstretchfactor\fontdimen3\font minus \fontdimen4\font\relax}
\providecommand{\BIBforeignlanguage}[2]{{%
\expandafter\ifx\csname l@#1\endcsname\relax
\typeout{** WARNING: IEEEtran.bst: No hyphenation pattern has been}%
\typeout{** loaded for the language `#1'. Using the pattern for}%
\typeout{** the default language instead.}%
\else
\language=\csname l@#1\endcsname
\fi
#2}}
\providecommand{\BIBdecl}{\relax}
\BIBdecl

\bibitem{hewing2020learning}
L.~Hewing, K.~P. Wabersich, M.~Menner, and M.~N. Zeilinger, ``Learning-based model predictive control: {Toward} safe learning in control,'' \emph{Annual Review of Control, Robotics, and Autonomous Systems}, vol.~3, no.~1, pp. 269--296, 2020.

\bibitem{kaufmann2023champion}
E.~Kaufmann, L.~Bauersfeld, A.~Loquercio, M.~M{\"u}ller, V.~Koltun, and D.~Scaramuzza, ``Champion-level drone racing using deep reinforcement learning,'' \emph{Nature}, vol. 620, no. 7976, pp. 982--987, 2023.

\bibitem{song2023reaching}
Y.~Song, A.~Romero, M.~M{\"u}ller, V.~Koltun, and D.~Scaramuzza, ``Reaching the limit in autonomous racing: Optimal control versus reinforcement learning,'' \emph{Science Robotics}, vol.~8, no.~82, p. eadg1462, 2023.

\bibitem{berberich2020data}
J.~Berberich, J.~K{\"o}hler, M.~A. M{\"u}ller, and F.~Allg{\"o}wer, ``Data-driven model predictive control with stability and robustness guarantees,'' \emph{IEEE Transactions on Automatic Control}, vol.~66, no.~4, pp. 1702--1717, 2020.

\bibitem{fazlyab2020safety}
M.~Fazlyab, M.~Morari, and G.~J. Pappas, ``Safety verification and robustness analysis of neural networks via quadratic constraints and semidefinite programming,'' \emph{IEEE Transactions on Automatic Control}, vol.~67, no.~1, pp. 1--15, 2020.

\bibitem{yan2023surviving}
Y.~Yan, X.-F. Wang, B.~J. Marshall, C.~Liu, J.~Yang, and W.-H. Chen, ``Surviving disturbances: {A} predictive control framework with guaranteed safety,'' \emph{Automatica}, vol. 158, p. 111238, 2023.

\bibitem{hsu2023safety}
K.-C. Hsu, H.~Hu, and J.~F. Fisac, ``The safety filter: A unified view of safety-critical control in autonomous systems,'' \emph{Annual Review of Control, Robotics, and Autonomous Systems}, vol.~7, 2023.

\bibitem{wabersich2023data}
K.~P. Wabersich, A.~J. Taylor, J.~J. Choi, K.~Sreenath, C.~J. Tomlin, A.~D. Ames, and M.~N. Zeilinger, ``Data-driven safety filters: Hamilton-jacobi reachability, control barrier functions, and predictive methods for uncertain systems,'' \emph{IEEE Control Systems Magazine}, vol.~43, no.~5, pp. 137--177, 2023.

\bibitem{ames2016control}
A.~D. Ames, X.~Xu, J.~W. Grizzle, and P.~Tabuada, ``Control barrier function based quadratic programs for safety critical systems,'' \emph{IEEE Transactions on Automatic Control}, vol.~62, no.~8, pp. 3861--3876, 2016.

\bibitem{wabersich2022predictive}
K.~P. Wabersich and M.~N. Zeilinger, ``Predictive control barrier functions: Enhanced safety mechanisms for learning-based control,'' \emph{IEEE Transactions on Automatic Control}, vol.~68, no.~5, pp. 2638--2651, 2022.

\bibitem{ha2025learning}
S.~Ha, J.~Lee, M.~van~de Panne, Z.~Xie, W.~Yu, and M.~Khadiv, ``Learning-based legged locomotion: State of the art and future perspectives,'' \emph{The International Journal of Robotics Research}, vol.~44, no.~8, pp. 1396--1427, 2025.

\bibitem{tearle2021predictive}
B.~Tearle, K.~P. Wabersich, A.~Carron, and M.~N. Zeilinger, ``A predictive safety filter for learning-based racing control,'' \emph{IEEE Robotics and Automation Letters}, vol.~6, no.~4, pp. 7635--7642, 2021.

\bibitem{jankovic2018robust}
M.~Jankovic, ``Robust control barrier functions for constrained stabilization of nonlinear systems,'' \emph{Automatica}, vol.~96, pp. 359--367, 2018.

\bibitem{cortez2022compatibility}
W.~S. Cortez and D.~V. Dimarogonas, ``On compatibility and region of attraction for safe, stabilizing control laws,'' \emph{IEEE Transactions on Automatic Control}, vol.~67, no.~9, pp. 4924--4931, 2022.

\bibitem{zinage2025universal}
V.~Zinage and E.~Bakolas, ``Universal barrier functions for safety and stability of constrained nonlinear systems,'' \emph{arXiv preprint arXiv:2511.01067}, 2025.

\bibitem{milios2024stability}
E.~Milios, K.~P. Wabersich, F.~Berkel, and L.~Schwenkel, ``Stability mechanisms for predictive safety filters,'' in \emph{IEEE Conference on Decision and Control}, 2024, pp. 2409--2416.

\bibitem{didier2024predictive}
A.~Didier, A.~Zanelli, K.~P. Wabersich, and M.~N. Zeilinger, ``Predictive stability filters for nonlinear dynamical systems affected by disturbances,'' \emph{IFAC-PapersOnLine}, vol.~58, no.~18, pp. 200--207, 2024.

\bibitem{rickert2014balancing}
M.~Rickert, A.~Sieverling, and O.~Brock, ``Balancing exploration and exploitation in sampling-based motion planning,'' \emph{IEEE Transactions on Robotics}, vol.~30, no.~6, pp. 1305--1317, 2014.

\bibitem{marcano2020review}
M.~Marcano, S.~D{\'\i}az, J.~P{\'e}rez, and E.~Irigoyen, ``A review of shared control for automated vehicles: Theory and applications,'' \emph{IEEE Transactions on Human-Machine Systems}, vol.~50, no.~6, pp. 475--491, 2020.

\bibitem{selvaggio2021autonomy}
M.~Selvaggio, M.~Cognetti, S.~Nikolaidis, S.~Ivaldi, and B.~Siciliano, ``Autonomy in physical human-robot interaction: A brief survey,'' \emph{IEEE Robotics and Automation Letters}, vol.~6, no.~4, pp. 7989--7996, 2021.

\bibitem{rawlings2020model}
J.~B. Rawlings, D.~Q. Mayne, and M.~Diehl, \emph{Model predictive control: {Theory}, computation, and design (2nd Edition)}.\hskip 1em plus 0.5em minus 0.4em\relax Nob Hill Publishing Madison, WI, 2017.

\bibitem{limon2005enlarging}
D.~Limon, T.~Alamo, and E.~F. Camacho, ``Enlarging the domain of attraction of {MPC} controllers,'' \emph{Automatica}, vol.~41, no.~4, pp. 629--635, 2005.

\bibitem{xie2021disturbance}
H.~Xie, L.~Dai, Y.~Lu, and Y.~Xia, ``Disturbance rejection {MPC} framework for input-affine nonlinear systems,'' \emph{IEEE Transactions on Automatic Control}, vol.~67, no.~12, pp. 6595--6610, 2021.

\bibitem{kim2025nonholonomic}
K.~H. Kim, V.~Todorovski, and M.~Krsti{\'c}, ``Nonholonomic robot parking by feedback--{Part II}: Nonmodular, inverse optimal, adaptive, prescribed/fixed-time and safe designs,'' \emph{arXiv preprint arXiv:2511.15219}, 2025.

\bibitem{nascimento2018nonholonomic}
T.~P. Nascimento, C.~E. D{\'o}rea, and L.~M.~G. Gon{\c{c}}alves, ``Nonholonomic mobile robots' trajectory tracking model predictive control: {A} survey,'' \emph{Robotica}, vol.~36, no.~5, pp. 676--696, 2018.

\end{thebibliography}


\end{document}